\documentclass[11pt,reqno]{article}
\usepackage[english]{babel}
\usepackage{amsmath,amsthm}
\usepackage{amsfonts}
\usepackage{graphicx}
\usepackage{enumerate}
\usepackage[usenames,dvipsnames]{color}
\usepackage{subfigure}
\usepackage[right,pagewise,displaymath, mathlines]{lineno}
\usepackage{epstopdf}
\usepackage{color}
\usepackage{multirow}
\usepackage[colorlinks=true,citecolor=blue]{hyperref}

\usepackage{natbib}
\usepackage{dsfont}
\usepackage{amssymb}

\newcommand{\E}{\mathbb{E}}
\newcommand{\p}{\mathbb{P}}
\newcommand{\R}{\mathbb{R}}
\newcommand{\X}{\mathcal X}

\renewcommand{\mid}{|}
\newcommand{\id}{\mathds{1}}

\def\d{\mathrm{d}}
\def\laweq{\buildrel \d \over =}

\newcommand{\VaR}{\mathrm{VaR}}
\newcommand{\VaRl}{\mathrm{Q}}
\newcommand{\ES}{\mathrm{ES}}

\numberwithin{equation}{section}
\numberwithin{figure}{section}
\numberwithin{table}{section}

\newtheorem{theorem}{Theorem}[section]

\newtheorem{lemma}{Lemma}[section]

\newtheorem{proposition}{Proposition}[section]

\theoremstyle{definition}
\newtheorem{definition}{Definition}[section]
\newtheorem{example}{Example}[section]

\theoremstyle{remark}

\newtheorem{remark}{Remark}[section]

\numberwithin{equation}{section}

\begin{document}

%
%

\begin{center}
{\Large\bf  Weak comonotonicity}

\vspace*{7mm}

{\large Ruodu Wang}

\smallskip

\textit{Department of Statistics and Actuarial Science, University of Waterloo, Waterloo, Ontario, N2L 5A7, Canada.}
\texttt{E-mail:~wang@uwaterloo.ca}

\bigskip

{\large Ri\v cardas Zitikis}

\smallskip

\textit{School of Mathematical and Statistical Sciences, University of
Western Ontario, London, Ontario N6A 5B7, Canada.}
\texttt{E-mail:~rzitikis@uwo.ca}

~

This version: July 2019

\end{center}

\medskip

\begin{quote}
\textbf{Abstract.}
The classical notion of comonotonicity has played a pivotal role when solving diverse problems in economics, finance, and insurance. In various practical problems, however, this notion of extreme positive dependence structure is overly restrictive and sometimes unrealistic. In the present paper, we put forward a notion of weak comonotonicity, which contains the classical notion of comonotonicity as a special case, and gives rise to necessary and sufficient conditions for a number of optimization problems, such as those arising in portfolio diversification, risk aggregation, and premium calculation.
In particular,
we show that a combination of weak comonotonicity and weak antimonotonicity with respect to some choices of measures is sufficient for the maximization of Value-at-Risk aggregation, and weak comonotonicity is necessary and sufficient for the Expected Shortfall aggregation.
Finally, with the help of weak comonotonicity acting as an intermediate notion of dependence between the extreme cases of no dependence and strong comonotonicity, we give a natural solution to a risk-sharing problem.

\smallskip

\end{quote}

\section{Introduction}
\label{sect-1}

Two functions are  said to be comonotonic if the ups and downs of one function follows those of the other function. Hence, though geometric in nature, comonotonicity is also a kind of dependence notion between functions. It is not surprising, therefore, that comonotonicity has given rise to sufficient conditions when solving a variety of problems in economics, banking, and insurance, and in particular those that deal with portfolio diversification, risk aggregation, and premium calculation principles. Our search for necessary and sufficient conditions has revealed that a certain augmentation of the classical (and inherently point-wise) notion of comonotonicity with appropriately constructed measures achieves more advanced goals than those associated with sufficient conditions. As a by-product, the augmented notion of comonotonicity, which we call weak comonotonicity, provides a natural bridge between a host of concepts in the aforementioned areas of application, and also in statistics, including measures of association. In what follows, we methodically  develop the notion of weak comonotonicity from first principles, establish its various properties, and demonstrate manifold uses.

Rigorously speaking, two functions $g$ and $h$ are comonotonic whenever the property
\begin{equation}\label{comon-0}
\big (g(x)-g(x')\big ) \big (h(x)-h(x')\big )\ge 0
\end{equation}
holds for all $x,x'\in \mathbb{R}$. This notion of comonotonicity \citep{S86} has played a pivotal role in sorting out numerous applications and developing new theories  \citep[e.g.,][]{Y87,D94}. Since then, these advances have been in the  mainstream of quantitative finance and economics literature \citep[e.g.,][]{DDGV02a, DDGV02b,FS16}. In this paper, we shall focus on dependence concepts between uni-dimensional functions (and random variables); for multivariate extensions and further references on comonotonicity, we refer to \cite{PS10}, \cite{CDG12}, \cite{EGH12}, and \cite{R13}. Note that if non-negativity in property~(\ref{comon-0}) is replaced by non-positivity, the functions $g$ and $h$ are said to be antimonotonic.

Comonotonicity of (Borel) functions $g$ and $h$ is a sufficient condition for non-negativity of the covariance $\mathrm{Cov}[g(X),h(X)]$, where $X$ is a random variable such that $g(X)$ and $h(X)$ have finite second moments.
This is immediately seen from the equations
\begin{align}
2\,\mathrm{Cov}[g(X),h(X)]
&=\mathbb{E}\big [ (g(X)-g(X'))(h(X)-h(X')) \big ]
\notag
\\
&=\iint_{\mathbb{R}^2} \big (g(x)-g(x')\big ) \big (h(x)-h(x')\big )
F_X(\mathrm{d}x)F_X(\mathrm{d}x') ,
\label{comon}
\end{align}
where $X'$ is an independent copy of $X$,
and  $F_X$ denotes the cumulative distribution function (cdf) of $X$.  The problem of determining the sign of covariances such as the one above has been of much interest in economics, insurance, banking, reliability engineering, and statistics. Several offshoots have arisen from this type of research, including quadrant dependence \citep{L66}, measures of association  \citep{EPW67}, monotonic \citep{KS78} and supremum \citep{G41} correlation coefficients.
The following example illustrates the need for such results.

\begin{example} \label{example-11}

Let $X$ be the severity of a risk, which could, for example, be a profit-and-loss variable. Let $g(X)$ be the cost associated with the risk $X$, and let $F_X^h$ be the so-called (knowledge-based) weighted cdf of the original random variable $X$ \citep[e.g.,][and references therein]{R97}. That is, $F_X^h$ is defined by the differential equation
\begin{equation}\label{weight-0}
F_X^h(\mathrm{d}x)={h(x)\over \mathbb{E}[h(X)]}F_X(\mathrm{d}x) ,
\end{equation}
where $h$ is a non-negative function such that $\mathbb{E}[h(X)]\in (0,\infty )$. The role of the function $h$ is to modify the probabilities of the original random variable $X$. For example, in insurance, it is usually designed to lower the left-hand tail of the pdf of $X$ and to lift its right-hand tail, thus making large insurance risks/losses more noticeable and the premiums loaded; we refer to, e.g., \cite{DG85} for the Esscher principle of insurance premium calculation, where $h(x)=e^{t x}$ for some constant $t>0$. Under the weighted cdf $F_X^h$, the average cost is
\[
\mathbb{E}^h[g(X)]=\int g(x)F_X^h(\mathrm{d}x) = {\mathbb{E}[g(X)h(X)]\over \mathbb{E}[h(X)]},
\]
which is not smaller than the average cost $\mathbb{E}[g(X)]$ under the true cdf $F_X$ if and only if the covariance $\mathrm{Cov}[g(X),h(X)]$ is non-negative. Several natural questions arise in this context: Under what conditions on the cost function $g$ and the probability weighting function $h$ is the covariance non-negative? Should the functions really be comonotonic, as our earlier arguments would suggest? It is important to note at this point that practical and theoretical considerations may or may not support the latter assumption, due to the complexity of economic agents' behaviour \citep[e.g.,][]{M52,PS03,GM09}.
\end{example}

We have organized the rest of the paper as follows. In Section~\ref{sect-wc}, we define, illustrate, and discuss the notion of weak comonotonicity, first for Borel functions and then for random variables (i.e., generic measurable functions). In Section~\ref{section-aggregation} we elucidate the role of weak comonotonicity in risk aggregation.
In particular,
we show that a combination of weak comonotonicity and weak antimonotonicity with respect to some sets of measures is sufficient for the maximization of Value-at-Risk (VaR) aggregation, and weak comonotonicity is necessary and sufficient for the Expected Shortfall (ES) aggregation. Both the VaR and the ES aggregation problems have been popular in the recent risk management literature \citep[e.g.,][]{R13, MFE15,EWW15}.
In Section~\ref{sect-indep}, we explore some properties of weak comonotonicity and its relation to other dependence structures and measures of association.  As most of this paper deals with weak comonotonicity with respect to product measures, in Section~\ref{sect-5} we illuminate the special role of these measures within the general context of joint measures. With the help of the developed theory, in Section~\ref{application} we present a detailed solution to a risk-sharing problem by invoking a weak comonotonicity constraint, whose naturalness becomes clear upon noticing that the assumption of arbitrary dependence among admissible allocations might sometimes be too weak, and the assumption of strong comonotonicity might be too strong, and so an intermediate dependence assumption based on weak comonotonicity arises most naturally. 
Section~\ref{sect-6} concludes the paper with a brief overview of main contributions.

\section{Weak comonotonicity}
\label{sect-wc}

Our efforts to tackle problems like those in the previous section, and in particular those related to risk aggregation (Section~\ref{section-aggregation}), have naturally led us to a notion of weak comonotonicity (to be defined in a moment) which naturally bridges the arguments around quantities in~(\ref{comon-0}) and (\ref{comon}) in the following way: First, note the equation
\begin{equation}\label{comon-0alternaive}
\big (g(x)-g(x')\big ) \big (h(x)-h(x')\big )
=\iint_{\mathbb{R}^2} \big (g(z)-g(z')\big ) \big (h(z)-h(z')\big )
\delta_{x}(\mathrm{d}z)\delta_{x'}(\mathrm{d}z') ,
\end{equation}
where $\delta_{x}$ and $\delta_{x'}$ are point masses at the points $x$ and $x'$, respectively. It now becomes obvious that by choosing various product measures instead of $\delta_{x}\times \delta_{x'}$, we can  seamlessly move from classical comonotonicity~(\ref{comon-0}) to covariance non-negativity~(\ref{comon}). Formalizing this flexibility gives rise to a general definition of weak comonotonicity, which is the topic of   Section~\ref{sect-wcbf}.

\subsection{Weak comonotonicity of Borel functions}
\label{sect-wcbf}

In what follows, we use $(\mathbb{R}, \mathcal{B})$ to denote the Borel measurable space, where $\mathcal B:= \mathcal B(\mathbb{R})$ is the Borel $\sigma$-algebra, and we also work with the measurable space $(\mathbb{R}^2, \mathcal{B}^2)$, where $\mathcal{B}^2:=\mathcal{B}\otimes \mathcal{B}$.

\begin{definition}\label{new-1}
Let $\mathcal R $ be any subset of product measures $\varrho_1\times \varrho_2$ on $(\mathbb{R}^2, \mathcal{B}^2)$. We say that two functions $g$ and $h$ are  \emph{weakly comonotonic with respect to $\mathcal R$} whenever
\begin{equation}\label{eq:def1}
\iint_{\mathbb{R}^2} \big (g(x)-g(x')\big ) \big (h(x)-h(x')\big )
\varrho_1(\mathrm{d}x)\varrho_2(\mathrm{d}x')\ge 0
\end{equation}
for every $\varrho_1\times \varrho_2\in \mathcal R$.
In case $\mathcal R$ is a singleton, we also say that $g$ and $h$ are  weakly comonotonic with respect to $\rho_1\times \rho_2$ if~\eqref{eq:def1} holds.
\end{definition}

We also speak of \emph{weak antimonotonicity} if non-negativity in~\eqref{eq:def1} is replaced by non-positivity.
Property~(\ref{eq:def1}) gives rise to a whole spectrum of comonotonicity notions, at one end of which is the classical notion of comonotonicity (i.e., property~(\ref{comon-0})), which can be viewed as $g$ and $h$ being weakly comonotonic with respect to $\mathcal R=\{  \delta_{x}\times \delta_{x'}:x,x' \in \mathbb{R}\}$. In other words, the classical notion of comonotonicity can be thought of as the point-wise or strong comonotonicity. On the other hand, Definition~\ref{new-1} and equation~(\ref{comon}) imply that the covariance $\mathrm{Cov}[g(X),h(X)]$ is non-negative if and only if the functions $g$ and $h$ are weakly comonotonic with respect to $\{F_X\times F_X\}$, where $F_X$ is the cdf of $X$. By choosing various product measures, we thus arrive at a large array of comonotonicity notions. The following example is designed to illustrate, and in particular enhance our intuitive understanding of, the notion of weak comonotonicity.

\begin{example} \label{example-21}
Let $g(x)=\sin(x)$ and $h(x)=\cos(x)$. In the classical sense, the two functions are neither comonotonic nor antimonotonic on the interval $[0,\pi ]$, but they are antimonotonic on $[0,\pi/2 ]$ and comonotonic on $[\pi/2,\pi ]$. As to their weak comonotonicity, consider the integral
\[
\Delta(a):=\iint_{\mathbb{R}^2} \big (g(x)-g(x')\big ) \big (h(x)-h(x')\big )
F(\mathrm{d}x)F(\mathrm{d}x')
\]
with respect to the following three uniform distributions $F=F_{[0,a]}$,  $F_{[(\pi-a)/2,(\pi+a)/2]}$, and $F_{[\pi-a,\pi]}$ on the noted intervals,
where $a\in [0,\pi ]$ in every case. We have
\[
\Delta(a)=
\left\{
  \begin{array}{ll}\displaystyle
    \frac{\sin ^2(a)}{ a}-\frac{2\sin (a) (1-\cos (a))}{a^2}
&\hbox{ when}\quad F=F_{[0,a]},
\\
  0
&\hbox{ when}\quad F=F_{[(\pi-a)/2,(\pi+a)/2]},
\\
    \displaystyle\frac{2\sin (a) (1-\cos (a))}{a^2}-\frac{\sin ^2(a)}{ a}
&\hbox{ when}\quad F=F_{[\pi-a,\pi]}.
  \end{array}
\right.
\]
When $F=F_{[\pi-a,\pi]}$, we depict $\Delta(a)$ as a function of $a\in [0,\pi ]$ in Figure~\ref{c-continous-right}.
\begin{figure}[h!]
\centering
\includegraphics[width=0.5\textwidth]{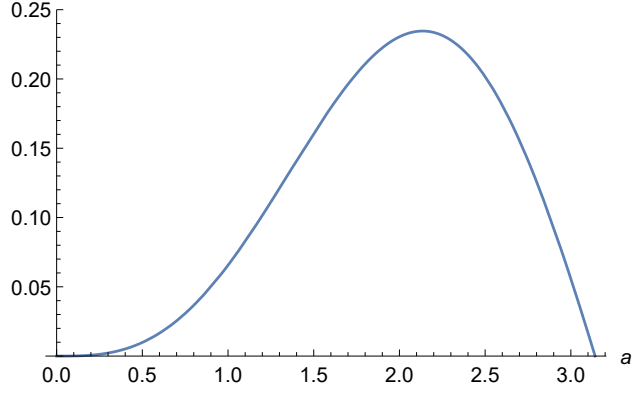}
\caption{Weak comonotonicity of $\sin(x)$ and $\cos(x)$ when $F$ is the uniform on $[\pi-a,\pi]$ distribution, depicted as the function $\Delta(a)$ for all  $a\in [0,\pi ]$.}
\label{c-continous-right}
\end{figure}
It is non-negative for every $a\in [0,\pi ]$, thus implying that the functions $\sin(x)$ and $\cos(x)$, which are neither comonotonic nor antimonotonic on $[0,\pi ]$ in the classical sense, are nevertheless \textit{weakly comonotonic} with respect to $\{F_{[\pi-a,\pi]}\times F_{[\pi-a,\pi]}:a\in [0,\pi ]\}$. On the other hand, when $F=F_{[0,a]}$, the function $\Delta(a)$ is non-positive for every $a\in [0,\pi ]$, and thus $\sin(x)$ and $\cos(x)$ are \textit{weakly antimonotonic} with respect to $\{F_{[0,a]}\times F_{[0,a]}:a\in [0,\pi ]\}$. Finally, under the distribution $F_{[(\pi-a)/2,(\pi+a)/2]}$, the two functions are both weakly comonotonic and weakly antimonotonic. This concludes Example~\ref{example-21}.
\end{example}

It is useful to reflect upon Example~\ref{example-21} from a general perspective, for which we employ Bayesian terminology. Namely, we first impose the (improper) uniform prior $\pi(x)\propto 1$ on the entire real line. Then we weight the prior using the indicator function $\mathbb{I}_{[x_0,x_1]}(x)$, where $[x_0,x_1]$ can be any compact interval. This gives rise to the uniform distribution $F_{[x_0,x_1]}$ defined by the differential equation
\begin{equation}\label{weight-1}
F_{[x_0,x_1]}(\mathrm{d}x)={\mathbb{I}_{[x_0,x_1]}(x)\over \mathbb{E}^{\pi}[\mathbb{I}_{[x_0,x_1]}]}\pi(\mathrm{d}x)
\end{equation}
(compare it with equation~\eqref{weight-0}). This uniform distribution, whose density (pdf) takes the form $f_{[x_0,x_1]}(x)=\mathbb{I}_{[x_0,x_1]}(x)/(x_1-x_0)$, can be thought of as a magnifying glass over the window $[x_0,x_1]$:  by sliding it over the domain of definition of functions, we explore weak comonotonicity of the functions, as we have done in Example~\ref{example-21}.

\subsection{Weak comonotonicity of random variables}
\label{sect-wcrv}

Note that the moment we had shifted our focus from non-decreasing functions to comonotonic ones, we lost the need for having order relationship in the underlying measurable space. Hence, we can work with abstract measurable space $(\Omega,\mathcal F)$, in which case $\mathcal{F}$-measurable functions like $X,Y  : \Omega \to \mathbb{R}$ are called random variables, and this is the general framework within which we work next. Namely, $X$ and $Y$ are said to be comonotonic whenever
\[
(X(\omega)-X(\omega') )(Y(\omega)-Y(\omega'))\ge 0
\]
for all $\omega,\omega'\in \Omega$. The definition is independent of any choice of measure.

\begin{definition}\label{new-10}
Let $\mathcal P $ be any subset of probability product measures $\pi_1\times \pi_2$ on $(\Omega^2 ,\mathcal F^2)$. We say that two random variables $X$ and $Y$ are  \emph{weakly comonotonic with respect to $\mathcal P$} whenever
\begin{equation}\label{eq:def2}
\iint _{\Omega^2}  (X(\omega)-X(\omega') )(Y(\omega)-Y(\omega')) \pi_1(\d \omega)\pi_2( \d \omega') \ge 0
\end{equation}
for every $\pi_1\times \pi_2\in \mathcal P$.
\end{definition}

Again, we also speak of \emph{weak antimonotonicity} if non-negativity in~\eqref{eq:def2} is replaced by non-positivity.
This definition not only generalizes Definition~\ref{new-1} but also paves a path toward the notion of conditional correlation, and thus, in turn, toward conditional beta that has prominently featured in problems such as dynamic asset pricing and risk estimation with non-synchronous prices \citep[see also references therein]{E16}. The next example elucidates the connection.

\begin{example}\rm \label{example-22}
Let $(\Omega ,\mathcal F, \p)$ be a probability space of financial scenarios $\omega \in \Omega $, and let $X,Y:\Omega \to \mathbb{R}$ be, for example, risk severities of two financial instruments. Quite often, it is of interest to measure association between the two instruments over certain events $A\in \mathcal F $ of positive probabilities. In this case, the original probability $\p $ is re-weighted
\[
\p (\d \omega \mid A)={\mathbb{I}_{A}(\omega) \over \p(A)}\p (\d \omega),
\]
thus reducing property~\eqref{eq:def2} via $\pi_1(\d \omega)=\pi_2(\d \omega)=\p (\d \omega \mid A)$ to
\begin{equation}\label{eq:def2-a}
\int _{A} \int _{A}  (X(\omega)-X(\omega') )(Y(\omega)-Y(\omega')) \p(\d \omega)\p( \d \omega') \ge 0 .
\end{equation}
Property~(\ref{eq:def2-a}) can in turn be rewritten as $\mathrm{Corr}[X,Y\mid A]\ge 0$, which can equivalently be interpreted as the non-negativity requirement on the conditional beta \citep{E16} over the events $A\in \mathcal F $ of interest, which could, for example, make up the $\sigma$-field of historical events (see, e.g., \cite{BJRL15} for a time series context; and \cite{PR07}, \cite{FS16} for risk measurement and management contexts).
\end{example}

Coming now back to Definition~\ref{new-10}, we check that the following four statements are equivalent:
\begin{enumerate}[(i)]
\item $X$ and $Y$ are (strongly, or point-wise) comonotonic;
\item $X$ and $Y$ are weakly comonotonic with respect to every probability product measures $\pi_1\times \pi_2$ on $(\Omega^2 ,\mathcal F^2)$;
\item $X$ and $Y$ are weakly  comonotonic with respect to
$\mathcal P=\{\delta_{\omega}\times \delta_{\omega'}:\omega,\omega'\in \Omega\}$;
\item there exist non-decreasing functions $f_1$ and $f_2$ and a random variable $Z$ such that $X=f_1(Z)$ and $Y=f_2(Z)$;
    according to Denneberg's Lemma \citep[Proposition 4.5]{D94}, we can set $Z:=X+Y$.
\end{enumerate}

We are now ready to elucidate the fundamental role of weak comonotonicity in problems associated with risk aggregation.

\section{Risk aggregation and weak comonotonicity}
\label{section-aggregation}

Two of the most popular classes of  risk measures used in banking and insurance practice are the Value-at-Risk (VaR)  and the Expected Shortfall (ES, also known as TVaR, CTE, CVaR, AVaR). We fix an atomless probability space $(\Omega,\mathcal F, \p)$.
For a random variable $X$, the VaR at  level $p\in(0,1)$  is defined as
$$
 \VaR_p(X)= \inf\{x\in \R: \p(X\le x)>  p\},
$$
and the ES at  level $p\in(0,1)$ is defined as
$$
  \ES_p(X)=\frac{1}{1-p}\int_p^1 \VaR_q(X)\d q.
$$

A classic problem in the field of risk management is risk aggregation with given marginal distributions \citep[e.g.,][Section 8.4]{MFE15}.   Let $X$ and $Y$ be two integrable random variables.
For $p\in (0,1)$, we say that $(X,Y)$ maximizes the $\VaR_p$ aggregation, if
$$\VaR_p(X+Y) =\max\{ \VaR_p(X'+Y'): X'\laweq X, ~Y'\laweq Y\},$$
and similarly for the ES aggregation, where ``$\laweq $'' stands for equality in distribution.

It is well-known  \citep[e.g.,][Section 8.4.4]{MFE15} that the maximization of ES aggregation is achieved by (strong) comonotonicity, that is, $(X,Y)$ maximizes the $\ES_p$ aggregation if they are strongly comonotonic.
A similar statement holds for all convex-order consistent risk measures, or variability measures, such as the variance, the standard deviation, convex and coherent risk measures, and the Gini Shortfall \citep{FWZ17}, and this is because of the well-known fact \citep[e.g.,][]{PW15} that comonotonicity maximizes convex order of the sum.
Note that for a specific $p\in (0,1)$, (strong) comonotonicity is a sufficient condition for $(X,Y)$ to maximize the $\ES_p$ aggregation,  but it is not necessary.

Another well-known phenomenon  \citep[e.g.,][Proposition 8.31]{MFE15}, which is in sharp contrast to the above situation, is that the maximization of VaR aggregation is not achieved by comonotonicity.
This is due to the fact that $\VaR_p$ is generally not subadditive.
The calculation of the worst-case VaR aggregation is technically very challenging and the corresponding dependence structure is quite complicated. For recent analytical and numerical results, we refer to \cite{WPY13} and \cite{EPR13, EPRWB14, EWW15}.
Fortunately, the case of $n=2$ admits an analytical solution, which is originally due to \cite{M81} and \cite{R82}.

To summarize, strong comonotonicity is sufficient but not necessary for the maximization of ES aggregation, and it is neither sufficient nor necessary for the maximization of VaR aggregation. This calls for weaker and alternative dependence notions compared to strong comonotonicity.
We shall see later in Theorem~\ref{th:var} that the notion of weak comonotonicity serves this purpose very well, as it gives a sufficient condition for the maximum $\VaR_p$ aggregation, as well as a necessary and sufficient condition for the maximum $\ES_p$ aggregation.

To prepare for Theorem~\ref{th:var}, we need some notation and a lemma. For a random variable $X$ and for any $p\in (0,1)$, we write
\[
A_p^X=\{\omega\in \Omega: X(\omega)>\VaR_p(X)\}.
\]
Note that $\p(A^X_p)= 1-p$ if $X$ is continuously distributed. In this case,
$A^X_p$ is the event of probability $p$ on which $X$ takes its largest possible values.
Further, let
\[
\mathcal P^X_p= \{\delta_\omega\times \delta_{\omega'}: \omega \in A^X_p, ~\omega' \in (A^X_p)^c\},
\]
where $A^c$ stands for the complement of a subset $A$ of $\Omega$,
and let
\[
{\mathcal Q}^X_p= \{\delta_\omega\times \delta_{\omega'}: \omega,\omega' \in A^X_p\}.
\]
In what follows, we treat $\p$-a.s.~equal random variables as identical, and thus
statements like ``$X$ and $Y$ are weakly comonotonic with respect to $\mathcal P^X_p$'' should be interpreted as they hold for a  representative pair of the random variables $X$ and $Y$.

\begin{lemma} \label{lem:var}
Let $X$ and $Y$ be two continuously distributed random variables, and let $p\in (0,1)$. The following three statements are equivalent:
\begin{enumerate}[\rm (i)]
\item  $X$ and $Y$ are weakly comonotonic with respect to $\mathcal P^X_p$;
\item  $X$ and $Y$ are weakly comonotonic with respect to $\mathcal P^Y_p$;
\item $A^X_p=A^Y_p$ a.s.~with respect to $\p $.
 \end{enumerate}
\end{lemma}

\begin{proof}
We only show (i)$\Leftrightarrow$(iii) since (ii)$\Leftrightarrow$(iii) holds by symmetry.
First, we assume that statement (i) holds.
For $\omega\in A^X_p$ and $\omega'\in (A^X_p)^c$, we have
$X(\omega)-X(\omega')>0$.
By definition of weak comonotonicity, this implies
$Y(\omega)-Y(\omega')\ge 0$.
Therefore, $Y$ takes its largest values on $A^X_p$.
Since $\p(A^X_p)=p$, we have
$A^Y_p=\{Y>\VaR_p(Y)\}=A^X_p$ a.s.
Next, we assume that statement (iii) holds. Then,
for a.s.~$\omega\in A^Y_p$ and $\omega'\in (A^Y_p)^c$, we have
$X(\omega)-X(\omega')>0$ and
$Y(\omega)-Y(\omega')>0$. This gives the weak comonotonicity of $X$ and $Y$; more precisely, of a representative version of $(X,Y)$.
\end{proof}

We are now ready to state our main result on the relationship between risk aggregation and weak comonotonicity.

\begin{theorem}\label{th:var}
Let $X$ and $Y$ be two continuously distributed and integrable random variables, and let $p\in (0,1)$. We have the following two statements:
\begin{enumerate}[\rm (i)]
\item  If  $X$ and $Y$ are weakly comonotonic with respect to $\mathcal P^X_p$, and $X$ and $Y$ are weakly antimonotonic with respect to $ {\mathcal Q}^X_p$,
 then $(X,Y)$ maximizes the $\VaR_p$ aggregation;
\item $X$ and $Y$ are weakly comonotonic with respect to $\mathcal P^X_p$ if and only if  $(X,Y)$ maximizes the $\ES_p$ aggregation.
   \end{enumerate}
\end{theorem}

\begin{proof}
First, we prove statement (i). By Lemma~\ref{lem:var},   $A^X_p=A^Y_p$ a.s.
Also note that  $X$ and $Y$ are (strongly) antimonotonic on the set $A^X_p$.
Let $U=F_X(X)$, which is uniformly distributed on $[0,1]$, and we know that $X$ and $U$ are strongly comonotonic.
As a consequence, $X=\VaR_{U}(X)$ a.s., and the sets $A^X_p$, $A^Y_p$ and  $\{U>p\}$ are a.s.~equal.
Because   $Y$  and $U$ are antimonotonic on the set $\{U>p\}$,
if $U$ takes value $u\in (p,1)$, then $Y$ takes the value
$\VaR_{1+p-u}(Y)$ a.s., and hence
 $Y=\VaR_{1+p-U}(Y)$ a.s.~on $\{U>p\}$.
Further, note that if $U\le p$, then $X+Y\le  \VaR_p(X)+\VaR_p(Y)$ a.s.~and if $U> p$, then  $X+Y\ge   \VaR_p(X)+\VaR_p(Y)$ a.s.
As a consequence, by definition of the $p$-quantile $(\VaR_p)$,  $\VaR_p(X+Y)$ is the smallest value ($\p$-a.s.) $X+Y$ takes on the set $\{U> p\}$, which is
the smallest value of $\VaR_{U}(X)+ \VaR_{1+p-U}(Y)$ for $U\in (p,1)$.
Therefore,
$$\VaR_p(X+Y)=
\inf \{\VaR_{p+t}(X)+\VaR_{1-t}(Y): t\in (0,1-p)\}.$$
This gives the maximum value of the $\VaR_p$ aggregation according to \citet[equation (2)]{M81} or \citet[Proposition 8.31]{MFE15}, thus concluding the proof of statement (i).

To prove statement (ii), we need some preliminaries. Namely, we use the dual representation of $\ES_p$ in the form
\begin{equation}\label{eq:ESrep} \ES_p(Z)=\max\{ \E[Z\mid B]: B\in \mathcal F,~\p(B)=1-p\}\end{equation}
for any random variable $Z$, and $B=A^Z_p$ attains the maximum in~\eqref{eq:ESrep} if  $Z$ is continuously distributed \citep[e.g.,][Lemma 3.1]{EW15}.
Because of subadditivity of $\ES_p$, we have
$$\ES_p(X)+\ES_p(Y)=\max\{ \VaR_p(X'+Y'): X'\laweq X, ~Y'\laweq Y\}.$$
Hence, $(X,Y)$ maximizes the $\ES_p$ aggregation if and only if $\ES_p(X+Y)=\ES_p(X)+\ES_p(Y)$.
Note that $\ES_p(X+Y)\le \ES_p(X)+\ES_p(Y)$ always holds. Now we are able to establish the ``if and only if'' statement (ii).
\begin{itemize}
  \item[($\Rightarrow$)]
Suppose that $X$ and $Y$ are weakly comonotonic with respect to $\mathcal P^X_p$.
This implies  $A^X_p=A^Y_p$ a.s.~by Lemma~\ref{lem:var}.
Therefore, by equation~\eqref{eq:ESrep},
\begin{align*}
 \ES_p(X+Y) \ge  \E\left[X+Y  \mid A^X_p\right]  =  \E\left[X \mid A^X_p\right] + \E\left[Y \mid A^Y_p\right] = \ES_p(X)+\ES_p(Y).
\end{align*}
Hence, $(X,Y)$ maximizes the $\ES_p$ aggregation.
  \item[($\Leftarrow$)]
Suppose that $(X,Y)$ maximizes the $\ES_p$ aggregation. Then, using equation~\eqref{eq:ESrep},  we have, for some $B\in\mathcal F$,
\begin{align*}
 \E\left[X+Y \mid B\right]= \ES_p(X+Y)= \ES_p(X)+\ES_p(Y)&= \E\left[X \mid A^X_p\right] + \E\left[Y \mid A^Y_p\right]\\&\ge \E\left[X \mid B\right]+ \E\left[Y \mid B\right].
\end{align*}
Therefore, $ \E [X \mid A^X_p ]  =  \E [X \mid B]$. Since $X$ is continuously distributed and takes its largest values on $A^X_p$, and $\p(A^X_p)=1-p=\p(B)$,
we conclude that $A^X_p=B$ a.s. Similarly, we conclude that $A^Y_p=B$ a.s. Using Lemma~\ref{lem:var} again, we obtain that $X$ and $Y$ are weakly comonotonic with respect to $\mathcal P^X_p$
\end{itemize}
This finishes the proof of Theorem~\ref{th:var}.
\end{proof}

Note that the weak comonotonicity condition on $\mathcal P^X_p$  in Theorem~\ref{th:var} is truly weaker than strong comonotonicity, as it does not specify the copula of $X$ and $Y$.
As discussed by \citet[Section~3]{EPRWB14}, the typical worst-case scenario of VaR aggregation is a combination of positive dependence and negative dependence in some non-rigorous sense.
Theorem~\ref{th:var}(i) answers precisely what these non-rigorous positive and negative dependence structures mean: weak comonotonicity with respect to $\mathcal P^X_p$
and weak antimonotonicity with respect to  $\mathcal Q^X_p$.
Furthermore, Theorem~\ref{th:var}(ii) gives a necessary and sufficient condition for the dependence structure maximizing the $\ES_p$ aggregation.

As a direct consequence of Theorem~\ref{th:var},  there exists a dependence structure that maximizes the $\VaR_p$ and $\ES_p$ aggregations simultaneously, as specified in Theorem~\ref{th:var}(i).
Note that the weak comonotonicity of $X$ and $Y$ with respect to $\mathcal P^X_p$ can be interpreted as a positive dependence in which the large values of $X$ and $Y$  appear simultaneously; but they are not perfectly aligned as in strong comonotonicity. It is straightforward to see, however, that this dependence structure, although necessary and sufficient for the $\ES_p$ aggregation, is not necessary for the $\VaR_p$ aggregation. For instance, if $Y$ is positive and $X(\omega)$ is large enough, say $X(\omega)>\VaR_p(X+Y)$, then it does not matter what value $Y(\omega)$ takes because it does not affect the calculation of $\VaR_p(X+Y)$.

\begin{remark}\rm
 Theorem~\ref{th:var}(ii) is formulated for a specific $p\in (0,1)$. If one likes $(X,Y)$ to maximize $\ES_p$ aggregation for all $p\in (0,1)$ or, equivalently, maximize the convex order of the sum, then strong comonotonicity is the only dependence structure
 \citep[e.g.,][Theorem 3]{C10}.
 This, in particular,  highlights the lack of practical attractiveness of the classical notion of comonotonicity, as it is unnecessarily too strong, at least from the perspective of $\ES_p$ aggregation. Indeed, practical considerations place emphasis on special values of $p$, usually specified by regulators, and they are, for example, close to 1 in banking and insurance (e.g., Basel IV and Solvency II; see \cite{MFE15}). More generally, we can think of examples when we would be concerned with $p$'s in certain subinterval of $(0,1)$, but  not in the entire interval $(0,1)$. This serves yet another justification for the introduction and explorations of the notion of weak comonotonicity.
\end{remark}

\begin{remark}\rm
The VaR aggregation  problem is equivalent to the problem of maximizing or minimizing $\p(X+Y>x)$ for a given $x\in \R$ and given marginal distributions of $X$ and $Y$. Indeed, this is the problem originally studied by \cite{M81} and \cite{R82}.
It has become well known since then that comonotonicity does not maximize or minimize the probability $\p(X+Y>x)$, and hence it is not the right notion to describe the corresponding dependence structures.
\end{remark}

\section{Some properties of weak comonotonicity}
\label{sect-indep}

In this section we explore some properties of weak comonotonicity, and its relation to notions of dependence structures and measures of association.


\subsection{Point-masses and comonotonicity}
\label{sect-31}

We have already noted that point masses reduce weak comonotonicity to strong comonotonicity, but the class
\[
\mathcal R_{g,h}=\big \{\rho_1\times \rho_2 : h \textrm{ and } g \textrm{ are weakly comonotonic with respect to } \rho_1\times \rho_2 \big \}
\]
depends, naturally, on the functions $g$ and $h$. In a sense, we can circumvent this dependence by introducing certain classes of point masses. Define
\[
\mathcal R_{\rm c}=\{\delta_x\times \delta_{x'} : x,x'\in \R\}
\]
and
\[
\mathcal R_{\rm a}=\{\delta_x\times \delta_x : x\in \R\}.
\]
Note that $\mathcal R_{g,h}$ is the largest set of product measures $\rho_1\times \rho_2$ with respect to which $g$ and $h $ are weakly comonotonic. The set $\mathcal R_{g,h}$ is never empty because $\mathcal R_{\rm a} \subseteq \mathcal R_{g,h}$. Finally, we note that for any two functions $g$ and $h$, the inclusions $\mathcal R_{\rm a}\subseteq \mathcal R_{g,h}$ and $\mathcal R_{\rm a} \subseteq \mathcal R_{\rm c}$ always hold.

\begin{theorem}\label{th-10}
We have the following two statements:
\begin{enumerate}[\rm (i)]
\item \label{th-10-a}
$\mathcal R_{g,h}\supseteq \mathcal R_{\rm c}$ if and only if $g$ and $h$ are strongly comonotonic.
\item \label{th-10-b}
$\mathcal R_{g,h}=\mathcal R_{\rm a}$ if and only if $g$ and $h$ are strongly antimonotonic and injective on $\R$.
\end{enumerate}
\end{theorem}

\begin{proof}
Statement~\eqref{th-10-a} is trivial. To prove statement~\eqref{th-10-b}, we first note that if $\mathcal R_{g,h}=\mathcal R_{\rm a}$,
then for any two $x,x'\in \R$ which are not identical,
we have $\delta_x\times \delta_{x'}\not \in \mathcal R_{g,h}$.
Thus, $(g(x)-g(x'))(h(x)-h(x'))<0$, and the desired injectivity and antimonotonicity follow.
 Next, assume  injectivity and antimonotonicity. Then,
 $(g(x)-g(x'))(h(x)-h(x'))<0$ for all $x,x'\in \R$ that are not identical.
For any product measure $\rho_1\times \rho_2$,
if condition~\eqref{eq:def1} holds, then
$\rho_1\times \rho_2$ must be supported in the points $(x,x')$ where either
$g(x)=g(x')$ or $h(x)=h(x')$,
and hence $x=x'$.
Since $\rho_1\times \rho_2$ is a product measure, we know that it has to be of the form $\delta_x\times \delta_x$ for $x\in \R$. This concludes the proof of Theorem~\ref{th-10}.
\end{proof}

We now turn our attention to random variables $X$ and $Y$. Similarly to $\mathcal R_{g,h}$, let
\[
\mathcal P_{X,Y}=\big \{\pi_1\times \pi_2 : X \textrm{ and } Y \textrm{ are weakly comonotonic with respect to } \pi_1\times \pi_2 \big \}.
\]
In other words, $\mathcal P_{X,Y}$ is the largest set of product measures with respect to which $X$ and $Y$ are weakly comonotonic.
It is a symmetric set with respect to $X$ and $Y$, that is, we have $\mathcal P_{X,Y}=\mathcal P_{Y,X}$. The validity of this symmetry easily follows from the equation
\begin{align}
\iint _{\Omega^2}  (X(\omega)&-X(\omega') )(Y(\omega)-Y(\omega')) \pi_1(\d \omega)\pi_2( \d \omega')
\notag
\\
&= \E^{\pi_1}[XY] +  \E^{\pi_2}[XY] - \E^{\pi_1}[X]\E^{\pi_2}[Y] - \E^{\pi_2}[X]\E^{\pi_1}[Y].
\label{eq:expect}
\end{align}
It also follows from the latter equation that if $\pi_1=\pi_2=:\pi $, then condition~\eqref{eq:def2} means that the correlation of $X$ and $Y$ under the measure $\pi$ is non-negative. Finally, we note that $\mathcal P_{X,Y}$ is invariant under all increasing linear marginal transforms, that is, the equation $\mathcal P_{\lambda_1 X+a_1,\lambda_2 Y+a_2} = \mathcal P_{X,Y}$ holds for all $\lambda _1,\lambda_2>0$ and $a_1,a_2\in \R$.

\begin{theorem}\label{th-11}
Let $\mathcal P_{\rm a}=\{\delta_\omega \times \delta_\omega : \omega \in \Omega\}$ and $\mathcal P_{\rm c}=\{\delta_\omega \times \delta_{\omega'} : \omega,\omega' \in \Omega\}$. We have the following two statements:
\begin{enumerate}[\rm (i)]
\item
$\mathcal P_{X,Y}\supseteq \mathcal P_{\rm c}$ if and only if $X$ and $Y$ are strongly comonotonic.
\item
$\mathcal P_{X,Y}=\mathcal P_{\rm a}$ if and only if $X$ and $Y$ are strongly antimonotonic and injective on $\Omega$.
\end{enumerate}
\end{theorem}

Note that $\mathcal P_{\rm a}\subseteq \mathcal P_{X,Y}$ and $\mathcal P_{\rm a}\subseteq \mathcal P_{\rm c}$. The proof of Theorem~\ref{th-11} is analogous to that of Theorem~\ref{th-10} and is therefore omitted.

\subsection{Set-masses and independence}
\label{sect-32}

We now go back to the integral, for  a probability space  $(\Omega,\mathcal F, \p)$,
$$
\iint _{\Omega^2}  (X(\omega)-X(\omega') )(Y(\omega)-Y(\omega')) \p(\d \omega)\p( \d \omega')
$$
 and distort, or rather weight, its probabilities. This gives rise to the integral
\begin{equation}\label{corr-probab}
\iint _{\Omega^2}  (X(\omega)-X(\omega') )(Y(\omega)-Y(\omega')) \p_{W_1}(\d \omega)\p_{W_2}( \d \omega'),
\end{equation}
where, for two random variables $W_1\ge 0$ and $W_2\ge 0$, the probability measure $\p_{W_1}$ is defined via the equation
$$
\p_{W_1}(\d \omega)={W_1(\omega) \over \mathbb{E}^{\p}[W_1]}\p(\d \omega),
$$
with $\p_{W_2}$ defined analogously. We next explore the case when the weights $W_1$ and $W_2$ are the indicators $\mathbb{I}_A$ and  $\mathbb{I}_B$, respectively, where $A$ and $B$ are elements of the $\sigma$-field $\mathcal F$.

Let $\sigma(X)$ denote the $\sigma$-field generated by $X$, and let
\[
\sigma^+(X)=\{A\in \sigma (X): \p(A)>0\}.
\]
For any event $A\in \sigma ^+(X)$, let $\p_A$ be the conditional probability of $\p$ on  $A$. We call these conditional probabilities set masses, which are natural extensions of the earlier explored point masses.

We shall next connect weak comonotonicity with (in)dependence of random variables $X$ and $Y$. It is instructive to start with the bivariate Gaussian case, and the following proposition is akin to the classical result which says that the equivalence of uncorrelatedness and independence characterizes Gaussian random variables.

\begin{proposition}\label{prop:gauss}
Let $(X,Y)$ be jointly Gaussian with standard margins and correlation $c\in [-1,1]$.  Then the following three statements are equivalent:
\begin{enumerate}[\rm (i)]
  \item $c\ge 0$;
  \item $\big \{\p_A\times \p_B: A, B\in \sigma^+(X)\big \} \subseteq \mathcal P_{X,Y}$;
  \item $\big \{\p_A\times \p_A: A\in \sigma^+(X)\big \} \subseteq \mathcal P_{X,Y}$.
\end{enumerate}
\end{proposition}

\begin{proof}
We first write $Y=cX+\sqrt{1-c^2}Z$ for some standard Gaussian $Z$ independent of $X$. For any $A\in \sigma^+(X)$, we have
$\E[XY \mid A]=\E[cX^2 \mid A]$ and $\E[Y \mid A]=\E[cX \mid A]$.
Therefore, the following holds if and only if $c\ge 0$:
$$
 \E[XY \mid A]   =c \E[X^2 \mid A] \ge c(\E[X \mid A])^2 = \E[X \mid A] \E[Y \mid A].
$$
Furthermore, we check that, for $c\ge 0$,
 \begin{align*}
 \E[XY \mid A]  + \E[XY \mid B]  - \E[X \mid A] & \E[Y \mid B] - \E[Y \mid A]\E[X \mid B] \\& = c \E[X^2 \mid A]  + c\E[X^2 \mid B] - 2c\E[X \mid A] \E[X \mid B]\\& \ge  c\left( \E[X^2 \mid A]  + \E[X^2 \mid B] - (\E[X \mid A])^2 -(\E[X \mid B])^2\right) \ge 0.
\end{align*}
This establishes the proposition.
\end{proof}

Generally, $\{\p_A\times \p_B: A, B\in \sigma^+(X)\} \subseteq \mathcal P_{X,Y}$ and $\{\p_A\times \p_A: A\in \sigma^+(X)\} \subseteq \mathcal P_{X,Y}$ are not equivalent conditions, although  they are in the Gaussian case, as we have just seen in Proposition~\ref{prop:gauss}.

\begin{proposition} We have the following statements:
\begin{enumerate}[\rm (i)]
  \item If $X$ and $Y$ are independent, then
$\{\p_A\times \p_B: A, B\in \sigma^+(X)\} \subseteq \mathcal P_{X,Y}$
and, by symmetry,
$\{\p_A\times \p_B: A, B\in \sigma^+(Y)\} \subseteq \mathcal P_{X,Y}$.
  \item If $\{\p_A\times \p_B: A, B\in \sigma^+(X)\} \subseteq \mathcal P_{X,Y}$, then, for $A,B\in \sigma^+(X)$, we have the property
\[
 \E[XY \mid A]  + \E[XY \mid B] - \E[X \mid A] \E[Y \mid B] - \E[Y \mid A]\E[X \mid B]\ge 0,
\]
which in the ``diagonal'' case $A=B$ reduces to non-negativity of the conditional correlation $\mathrm{Corr}[X,Y\mid A]$ for every event $A\in \sigma^+(X)$.
\end{enumerate}
\end{proposition}

\begin{proof}
To prove part (i), we use equation~\eqref{eq:expect} and have
  \begin{align*}
&\iint _{\Omega^2}  (X(\omega)-X(\omega') )(Y(\omega)-Y(\omega')) \p_A(\d \omega)\p_B( \d \omega')
\\
&=  \E[XY \mid A]  + \E[XY \mid B] - \E[X \mid A] \E[Y \mid B] - \E[Y \mid A]\E[X \mid B]
\\
&=  \E[X \mid A] \E[Y]  + \E[X \mid B] \E[Y] - \E[X \mid A] \E[Y ] - \E[Y ]\E[X \mid B] = 0.
\end{align*}
Hence
$\{\p_A\times \p_B: A, B\in \sigma^+(X)\} \subseteq \mathcal P_{X,Y}$. The other half of (i) is by symmetry. The proof of statement (ii) is a straightforward verification.
\end{proof}

\subsection{Weak comonotonicity and measures of association}
\label{sect-assoc}

The notion of weak comonotonicity has enabled us to establish a whole spectrum of comonotonicity notions, ranging from the classical (strong) comonotonicity under the pairs of all point masses to weaker comonotonicity notions under the pairs of more elaborate measures. As we shall see next, this flexibility enables us to capture a whole array of measures of association.

\begin{enumerate}[(S1)]
\item\label{S1}
The Pearson correlation $\mathrm{Corr}(X,Y)$ is non-negative if and only if $X$ and $Y$ are weakly comonotonic with respect to  $\mathbb{P}\times \mathbb{P}$.

\item\label{S2}
Two random variables $X$ and $Y$ are positively associated (also called positively function dependent; see \cite{J97} for details) if and only if for all non-decreasing functions $h$ and $g$, the random variables $h(X)$ and $g(Y)$ are  weakly comonotonic with respect to  $\mathbb{P}\times \mathbb{P}$.

\item\label{S3}
Assuming that $X$ and $Y$ have continuous cdf's $F_X$ and $F_Y$, respectively,  the Spearman correlation is non-negative if and only if  $F_X(X)$ and $F_Y(Y)$ are  weakly comonotonic with respect to the product $\mathbb{P}\times \mathbb{P}$.

\item\label{S4}
Two random variables $X$ and $Y$ are independent if and only if, for all $A, B\in \mathcal B$, the indicators $\mathbb{I}_{\{X\in A\}}$  and $\mathbb{I}_{\{Y\in B\}}$ are weakly comonotonic with respect to  $\mathbb{P}\times \mathbb{P}$.  The same statement holds if we replace weak comonotonicity by weak antimonotonicity.
\end{enumerate}


All the above statements are straightforward and follow from the equivalence of weak comonotonicity (with respect to  $\mathbb{P}\times \mathbb{P}$) and covariance non-negativity. The fourth property, however, warrants a simple comment-like proof.

\begin{proof}[Proof of (S\ref{S4})] It is obvious that independence implies  weak comonotonicity, as well as weak antimonotonicity, of $\mathbb{I}_{\{X\in A\}}$  and $\mathbb{I}_{\{Y\in B\}}$.
For the other direction,
let $(X',Y')$ be an independent copy of $(X,Y)$. For all $A, B\in \mathcal B$, we have
\[
\mathbb{E}[ (\mathbb{I}_{\{X\in A\}} - \mathbb{I}_{\{X'\in A\}} ) (\mathbb{I}_{\{Y\in B\}} - \mathbb{I}_{\{Y'\in B\}})]
 = 2\mathbb{P}(X\in A, Y\in B)- 2\mathbb{P}(X\in A)\mathbb{P}(Y\in B),
\]
which is non-negative. Likewise, we have
\[
\mathbb{E}[ (\mathbb{I}_{\{X\in A\}} - \mathbb{I}_{\{X'\in A\}} ) (\mathbb{I}_{\{Y\in B^c\}} - \mathbb{I}_{\{Y'\in B^c\}})]\\
= 2\mathbb{P}(X\in A, Y\in B^c)- 2\mathbb{P}(X\in A)\mathbb{P}(Y\in B^c ),
\]
which is also non-negative. Adding the left-hand sides of the two equations gives zero, which, due to the just established non-negativity statements, implies that the right-hand sides are also zeros, which implies independence.
\end{proof}

It is convenient to have probability-based quantities expressed in terms of distribution functions, and we next do so expressly for the purpose of checking whether or not the random variables $h(X)$ and $g(Y)$ are  weakly comonotonic with respect to  $\mathbb{P}\times \mathbb{P}$. To this end, we write the equations
\begin{align}\label{av-comon-4a}
\iint_{\Omega^2}
\big (g(X(\omega))-g(X(\omega'))\big )
& \big (h(Y(\omega))-h(Y(\omega'))\big )
\mathbb{P}(\mathrm{d} \omega) \mathbb{P}(\mathrm{d} \omega')
\notag
\\
&=\mathbb{E}\big [ (g(X)-g(X'))(h(Y)-h(Y')) \big ]
\notag
\\
&= \mathbb{E}\big [ (g(X)-g(X'))(h^*(X)-h^*(X')) \big ]
\notag
\\
&=\iint_{\mathbb{R}^2} \big (g(x)-g(x')\big ) \big (h^*(x)-h^*(x')\big )
F_X(\mathrm{d}x)F_X(\mathrm{d}x'),
\end{align}
where
\[
h^*(x):=\mathbb{E}\big [ h(Y) \mid X=x \big ].
\]
Consequently, $h(X)$ and $g(Y)$ are  weakly comonotonic with respect to  $\mathbb{P}\times \mathbb{P}$ if and only if the functions $g$ and $h^*$ are weakly comonotonic with respect to $F_X\times F_X$, that is,
\begin{equation}\label{av-comon-4}
\iint_{\mathbb{R}^2} \big (g(x)-g(x')\big ) \big (h^*(x)-h^*(x')\big )
F_X(\mathrm{d}x)F_X(\mathrm{d}x')\ge 0 .
\end{equation}
From this we arrive at the following interpretation of positive association in terms of weak comonotonicity.

\begin{proposition}\label{th-2a}
The following two statements are equivalent:
\begin{enumerate}[\rm (1)]
 \item \label{stat-1a}
The random variables $X$ and $Y$ are positively associated.
  \item \label{stat-2a}
For all non-decreasing Borel functions $g$ and $h$, the functions $g$ and $h^*(x):=\mathbb{E} [ h(Y) \mid X=x ]$ are weakly comonotonic with respect to $F_X\times F_X$.
\end{enumerate}
\end{proposition}

From Proposition~\ref{th-2a} we see that if we require the functions $g$ and $h^*$ to be weakly comonotonic with respect to \textit{all} product measures $\varrho_1\times \varrho_2$, and thus in particular with respect to the products $\delta_{x}\times \delta_{x'}$ for all $x,x' \in \mathbb{R}$, then this is tantamount to the functions $g$ and $h^*$ being strongly comonotonic. The next theorem connects the notion of weak comonotonicity of $g$ and $h^*$ with the notion of positive regression dependence \citep{L66}.

\begin{proposition}\label{th-3b}
The following two statements are equivalent:
\begin{enumerate}[(i)]
  \item \label{stat-1}
For all non-decreasing Borel functions $g$ and $h$, the functions $g$ and $h^*$ are weakly comonotonic with respect to all product measures $\varrho_1\times \varrho_2$.
  \item  \label{stat-2}
The random variable $Y$ is positively regression dependent on $X$, that is, for every $y\in \mathbb{R}$, the function $x\mapsto F_{Y\mid X}(y\mid x)$ is non-increasing.
\end{enumerate}
\end{proposition}

\begin{proof}
Statement~(\ref{stat-1}) means that $g$ and $h^*$ are strongly comonotonic for all non-decreasing Borel functions $g$ and $h$. With this in mind, the equivalence of statements~(\ref{stat-1}) and (\ref{stat-2}) follows by noting that $h^*(x)$ and $1-F_{Y\mid X}(y\mid x)$ are equal to $\mathbb{E}[h(Z_x)]$ and $\mathbb{E}[h_y(Z_x)]$, respectively, where $Z_x:=[Y\mid X=x]$ and $h_y=\mathbb{I}_{(y,\infty)}$. It now remains to recall that the class of all non-decreasing functions $h$ and the class $\{h_y$, $y\in \mathbb{R}\}$ give rise to two equivalent ways for defining stochastic ordering \citep[e.g.,][]{PR07, R13, FS16}.
\end{proof}

%
%
%

\section{Maximality of product measures}
\label{sect-5}

Definition~\ref{new-10} is based on the product measure $\pi_1 \times \pi_2$, which is a natural choice in view of the examples that have given rise to the notion of weak comonotonicity. There are, however, situations when the need for more generality arises, and for this we introduce an extension of integral~\eqref{eq:def2}:
\begin{equation}\label{corr-probab-1}
\iint _{\Omega^2}  (X(\omega)-X(\omega') )(Y(\omega)-Y(\omega')) \pi_{W}(\d \omega,\d \omega'),
\end{equation}
where, $\pi$ is a measure on $(\Omega,\mathcal F)$, and for any random variable $W$ on $(\Omega^2 ,\mathcal F^2)$,
$$
\pi_{W}(\d \omega,\d \omega')={W(\omega,\omega')\over \mathbb{E}^{\pi\times \pi}[W]}\pi(\d \omega)\pi(\d \omega').
$$

\begin{definition}\label{new-11}
We say that random variables $X$ and $Y$  are  weakly comonotonic with respect to a set $\mathcal P$ of (not necessarily product) measures $\pi $ on $(\Omega^2 ,\mathcal F^2)$ whenever
\[
 \iint _{\Omega^2}  (X(\omega)-X(\omega') )(Y(\omega)-Y(\omega')) \pi(\mathrm{d} \omega,\mathrm{d} \omega') \ge 0
\]
for all $\pi\in \mathcal P$.
\end{definition}

This generalization provides a context within which we can better understand the role of the product measure $\pi_1 \times \pi_2$, which happens to enjoy the following maximality property:
\begin{multline}
\iint _{\Omega^2}  (X(\omega)-X(\omega') )(Y(\omega)-Y(\omega')) \pi(\mathrm{d} \omega,\mathrm{d} \omega')
\\
\le \iint _{\Omega^2}  (X(\omega)-X(\omega') )(Y(\omega)-Y(\omega')) \pi_1(\mathrm{d} \omega) \pi_2 (\mathrm{d} \omega')  ,
\label{qqq-1}
\end{multline}
provided that
\begin{multline}
\mathbb{C}^{\pi}(X,Y):={1\over 2}\left\{\iint _{\Omega^2}  X(\omega)Y(\omega') \pi(\mathrm{d} \omega,\mathrm{d} \omega')
-\int _{\Omega}  X(\omega) \pi_1(\mathrm{d} \omega)
\int _{\Omega}  Y(\omega')\pi_2 (\mathrm{d} \omega') \right\}
\\
+{1\over 2}\left\{\iint _{\Omega^2}  Y(\omega)X(\omega') \pi(\mathrm{d} \omega,\mathrm{d} \omega')
-\int _{\Omega}  Y(\omega) \pi_1(\mathrm{d} \omega)
\int _{\Omega}  X(\omega')\pi_2 (\mathrm{d} \omega')\right\} \ge 0,
\label{qqq-1star}
\end{multline}
where $\pi_1(A):=\int_{\Omega} \pi(A,\mathrm{d} \omega')$ and
$\pi_2(A'):=\int_{\Omega} \pi(\mathrm{d} \omega,A')$. If the measure $\pi $ is symmetric, that is, $\pi(A,A')=\pi(A',A)$ for all $A,A'\in \mathcal{F}$, then $\pi_1=\pi_2$. Note also that the covariance-looking quantities inside the first braces and inside the second braces are \textit{not}, in general, symmetric with respect to $X$ and $Y$, but their sum $\mathbb{C}^{\pi}(X,Y)$ is always symmetric, irrespective of the measure $\pi $. Finally, we note that in the ``diagonal'' case $X=Y$, we have
\begin{equation*}
\mathbb{C}^{\pi}(X,X)=\iint _{\Omega^2}  X(\omega)X(\omega') \pi(\mathrm{d} \omega,\mathrm{d} \omega')
-\int _{\Omega}  X(\omega) \pi_1(\mathrm{d} \omega)
\int _{\Omega}  X(\omega')\pi_2 (\mathrm{d} \omega') .
\end{equation*}

To get a deeper insight into the above notion, and to also connect it to weak  comonotonicity and positive association, we shift our focus to 1) the measurable space $(\mathbb{R}^2,\mathcal{B}^2)$, 2)  Borel functions $g$ and $h$, and 3) the joint cdf $F_{V,W}$ generated by two random variables $V$ and $W$, whose marginal cdf's we denote by $F_{V}$ and $F_{W}$, respectively. Under this scenario, bound~(\ref{qqq-1}) takes on the following form
\begin{multline}\label{eqq-0}
\iint_{\mathbb{R}^2} \big (g(v)-g(w)\big ) \big (h(v)-h(w)\big )F_{V,W}(\mathrm{d}v,\mathrm{d}w)
\\
\le \iint_{\mathbb{R}^2} \big (g(v)-g(w)\big ) \big (h(v)-h(w)\big )F_V(\mathrm{d}v)F_W(\mathrm{d}w),
\end{multline}
which holds (cf.~condition~(\ref{qqq-1star})) if and only if
\begin{equation}\label{eqq-2}
\mathbb{C}^{\pi}(g,h):={1\over 2}\mathrm{Cov}[g(V),h(W)]+{1\over 2}\mathrm{Cov}[h(V),g(W)]\ge 0,
\end{equation}
where $\pi=F_{V,W}$. Obviously, $\mathbb{C}^{\pi}(g,h)=\mathbb{C}^{\pi}(h,g)$ irrespective of the measure $\pi$, and we also have the equation $\mathbb{C}^{\pi}(g,g)=\mathrm{Cov}[g(V),g(W)]$.

From the above notes we conclude that within the class of measures $\pi=F_{V,W}$ generated by positively-associated random variables $V$ and $W$, the product measure $\pi_0=F_V\times F_V$ is maximal in the sense of bound~(\ref{eqq-0}) within the class of all pairs of non-decreasing Borel functions $g$ and $h$. But the assumptions that 1) $V$ and $W$ are positively associated and 2) $g$ and $h$ are non-decreasing are rather strong: they ensure non-negativity of the two covariances on the right-hand side of equation~(\ref{eqq-2}) and thus, in turn, imply the required non-negativity of $\mathbb{C}^{\pi}(g,h)$.

Due to the notion of weak comonotonicity, we can specify necessary and sufficient conditions for non-negativity of the two covariances on the right-hand side of equation~(\ref{eqq-2}). For this, we write
\begin{equation}\label{eqq-4}
\mathbb{C}^{\pi}(g,h)={1\over 2}\mathrm{Cov}[g(V),h^*(V)]+{1\over 2}\mathrm{Cov}[g^*(V),h(V)],
\end{equation}
where $h^*(v)=\mathbb{E}[ h(W) \mid V=v ]$ and $g^*(v)=\mathbb{E}[ g(W) \mid V=v ]$. The two covariances on the right-hand side of equation~(\ref{eqq-4}) are non-negative if and only if the two pairs $(g,h^*)$ and $(g^*,h)$ are weakly comonotonic with respect to the measure $\pi_0=F_V\times F_V$.

Note, however, that the covariance $\mathbb{C}^{\pi}(g,h)$ can be non-negative without making the two covariances on the right-hand side of equation~(\ref{eqq-4}) non-negative. To show this, we next construct an example when one of the two covariances is negative but $\mathbb{C}^{\pi}(g,h)$ is positive.

\begin{example} \label{example-51}
Let $g(x)=\sin(x)$ and $h(x)=\cos(x)$. Furthermore, let $V$ and $W$ be random variables whose marginal distributions are
\[
V=
\left\{
  \begin{array}{ll}
    0     & \hbox{with }~ 3/10\\
    \pi/2 & \hbox{with }~ 7/10
  \end{array}
\right.
\]
and
\[
W=
\left\{
  \begin{array}{ll}
    2\pi/3     & \hbox{with }~ 3/10\\
    \pi & \hbox{with }~ 7/10
  \end{array}
\right.
\]
and let the dependence structure be given by the matrix
\[
\bordermatrix{%
      & 2\pi/3 & \pi \cr
    0 & 1/10 & 2/10 \cr
\pi/2 & 2/10 & 5/10 \cr
}
\]
with Archimedes' constant $\pi\approx 3.14159$ not be confused with the earlier used notation for measures. We have
\begin{gather*}
\mathrm{Cov}[g(V),h(W)]=-{1\over 200}=-0.005,
\\
\mathrm{Cov}[h(V),g(W)]={\sqrt{3}\over 200}\approx 0.00866,
\end{gather*}
and thus
\[
\mathbb{C}^{\pi}(g,h)={1\over 2}\bigg( -{1\over 200}+{\sqrt{3}\over 200}\bigg)\approx 0.00183.
\]
This concludes Example~\ref{example-51}.
\end{example}

\section{An application to quantile-based risk sharing}
\label{application}

In this section, we illustrate the above developed theory by studying an optimization problem arising in the context of risk sharing, where weak comonotonicity provides a natural constraint on the dependence structure of  admissible risk allocations.
We follow the framework of
\cite{ELW18, ELMW19}, who studied risk sharing problems with quantile-based risk measures.


Let $\X$ be the set of all random variables in an atomless probability space.
 The random variable $X\in \X$ represents a total random loss, and $\rho_1,\dots,\rho_n$ are risk measures (e.g., VaR or ES) used by $n$ economic agents (e.g., firms or investors).
Denote
\begin{align}
\mathbb{A}_n(X)=\left\{(X_1,\ldots,X_n)\in \mathcal{X}^n: \sum_{i=1}^nX_i\ge X\right\} , \label{eq:rs1}
\end{align}
which is the set of all possible allocations of losses to the agents, summing up to at least the total loss $X$.
By \citet[Proposition 1]{ELW18}, Pareto-optimal allocations for the risk sharing problem are solutions to the following optimization problem
\begin{align}
\min\left\{\sum_{i=1}^n\rho_i(X_i): (X_1,\dots,X_n)\in \mathbb{A}_n(X) \right\}.  \label{eq:rs2}
\end{align}
In   problem~\eqref{eq:rs2},
the dependence structure among the allocation $(X_1,\dots,X_n)$  is arbitrary.
 \cite{ELW18} also consider the constrained problem
\begin{align}
\min\left\{\sum_{i=1}^n\rho_i(X_i): (X_1,\dots,X_n)\in \mathbb{A}_n(X),~X_i\uparrow X,~ i=1,\ldots,n  \right\}. \label{eq:rs3}
\end{align}
where $X_i\uparrow X$ means that $X_i$ and $X$ are strongly comonotonic.

For a practical situation, the assumption of arbitrary dependence in the admissible allocations as in problem~\eqref{eq:rs2} may be too weak,
and the assumption of strongly comonotonic  allocations  in problem~\eqref{eq:rs3} may be too strong.
Therefore, we can consider an intermediate assumption on the dependence structure of the admissible allocations in the risk sharing problem,
which is modelled by weak comonotonicity.

To this end, we construct a spectrum of weak comonotonicity indexed by $\beta \in [0,1]$, such that $\beta =0$ corresponds to no dependence constraint and $\beta=1$ corresponds to strong comonotonicity.
For this purpose,
recall that in Section~\ref{section-aggregation} above, for a random variable $X$ and for any $p\in [0,1)$, we defined
 $$
A_p^X=\{\omega\in \Omega: X(\omega)>\VaR_p(X)\}
$$
and
 $$
\mathcal P^X_p= \{\delta_\omega\times \delta_{\omega'}: \omega \in A^X_p, ~\omega' \in (A^X_p)^c\}.
 $$
In what follows, for two random variables $Y$ and $Z$,
we shall use the notation $Y\uparrow_\beta Z$ when $Y$ and $Z$ are weakly comonotonic with respect to $ \bigcup_{p\in [ 1-\beta,1)} \mathcal P^Z_p$.

The interpretation of $Y\uparrow_\beta Z$ is that $Y$ and $Z$ are comonotonic and both take large values on the event $A_{1-\beta}^Z$, and there is no dependence assumption on $(A_{1-\beta}^Z)^{c}$. Note also that the requirement  $Y\uparrow_\beta Z$  gets stronger when $\beta$ increases. In particular,
assuming that $Z$ is continuously distributed,
  for $\beta=0$, $Y\uparrow_\beta Z$
imposes no dependence assumption, and for $\beta=1$, it means that $Y$ and $Z$ are strongly comonotonic.
Using this connection, we will impose   $X_i\uparrow_\beta X$, $i=1,\dots,n$ as a constraint on the admissible allocations in our risk sharing problem,
so that $\beta=0$ corresponds to \eqref{eq:rs2} and $\beta=1$ corresponds to \eqref{eq:rs3}.

For the purpose of illustration, we focus on an important special case studied by \cite{ELW18}, when the risk measures $\rho_1,\dots,\rho_n$  are quantiles at different levels.
Following the setup of \cite{ELW18}, for $\alpha \in (0,1)$ and $Y\in \X$, we define
$$
\VaRl_\alpha(Y) = \inf \{x\in \R: \p(Y\le x) \ge 1-\alpha \}.
$$
\begin{remark}
Note that   $\VaRl_\alpha$ is the left $(1-\alpha)$-quantile, which is different from the VaR (right quantile) defined in Section~\ref{section-aggregation}.
The choice of the left quantile here and in \cite{ELW18,ELMW19}  is intentional.
For minimization problems, we need to work with left quantiles to guarantee the existence of optimal allocations. Recall that in Section~\ref{section-aggregation} we study maximization problems, and hence   right quantiles are natural choices there.
On the other hand, using $(1-\alpha)$-quantile instead of $\alpha$-quantile  leads to concise statements of the results; this will be clear from statements~\eqref{eq:elw1}--\eqref{eq:elw2} below.
\end{remark}

Let $\rho_i=\VaRl_{\alpha_i}$, $i=1,\dots,n$, where $\alpha_1,\dots,\alpha_n$ are positive constants such that $ \sum_{i=1}^n \alpha<1$.
For this choice of risk measures,  both problems~\eqref{eq:rs2} and~\eqref{eq:rs3} admit analytical solutions, given in Theorem 2 and Proposition 5 of \cite{ELW18}, respectively. These results imply
\begin{align}
\min\left\{\sum_{i=1}^n\VaRl_{ \alpha_i} (X_i): (X_1,\dots,X_n)\in \mathbb{A}_n(X) \right\} = \VaRl_{ \sum_{i=1}^n \alpha_i}(X)
\label{eq:elw1}
\end{align}
and
\begin{align}
\min\left\{\sum_{i=1}^n\VaRl_{ \alpha_i} (X_i): (X_1,\dots,X_n)\in \mathbb{A}_n(X), ~X_i\uparrow X,~ i=1,\ldots,n  \right\} = \VaRl_{ \bigvee_{i=1}^n \alpha_i}(X),
\label{eq:elw2}
\end{align}
and the corresponding optimal allocations can be explicitly constructed as well.
Note that result~\eqref{eq:elw1} implies
\begin{align}
\VaRl_{ \sum_{i=1}^n \alpha_i} \left(\sum_{i=1}^nX_i\right) \le
 \sum_{i=1}^n\VaRl_{ \alpha_i} (X_i) \label{eq:elw3}\end{align}
 for all $X_1,\dots,X_n \in \X$ \citep[Corollary 1]{ELW18}, which will be useful in our analysis below.

\begin{remark}
 \cite{ELW18} formulate the admissible allocations in ~\eqref{eq:rs1} using  $\sum_{i=1}^nX_i= X$ instead of
$\sum_{i=1}^nX_i\ge X$. It is easy to see that in problems~\eqref{eq:rs2} and \eqref{eq:rs3}, these two setups are equivalent for monotone risk measures such as the quantiles.
In this paper, we use inequality in definition~\eqref{eq:rs1} because our dependence constraint would make the two formulations generally no longer equivalent, and analytical solutions are found for the current formulation.
\end{remark}


For a continuously distributed $X$ and a parameter $\beta \in [0,1]$, we consider the optimization problem
\begin{align}
V_\beta (X)=\inf\left\{\sum_{i=1}^n\VaRl_{ \alpha_i}(X_i): (X_1,\dots,X_n)\in \mathbb{A}_n(X),~X_i\uparrow_\beta X,~ i=1,\ldots,n  \right\}. \label{eq:rs4}
\end{align}
It is clear that  $\beta=0$ corresponds to problem~\eqref{eq:rs2} and $\beta=1$ corresponds to problem~\eqref{eq:rs3}.
Therefore, the use of weak comonotonicity   yields   a bridge between the two risk sharing problems~\eqref{eq:rs2} and~\eqref{eq:rs3} considered by \cite{ELW18},
and it offers more flexibility as one can impose a partial dependence constraint on the admissible allocations.

Similarly to many other optimization problems involving quantiles (or VaR), problem~\eqref{eq:rs4} is not convex   as $\VaRl_\alpha$  is generally not convex, and thus a specialized analysis of the problem is needed. Nevertheless, via some auxiliary  technical results, we will show below that problem~\eqref{eq:rs4}   admits an analytical solution, and an optimal allocation will be obtained in explicit form.

 \begin{theorem}\label{th:rs}
Suppose that  $X$ is a continuously distributed random variable,  $\alpha_1,\dots,\alpha_n>0$, $\sum_{i=1}^n \alpha_i<1$, and $\beta \in [0,1]$.
We have
 \begin{align*}
V_\beta (X)= \VaRl_{ \gamma}(X),
\end{align*}
where $\gamma =  \beta \wedge (\bigvee_{i=1}^n  \alpha_i)+ \sum_{i=1}^n (\alpha_i-\beta)_+. $
 %
 \end{theorem}

\begin{proof}
We first note that $\gamma = \sum_{i=1}^n \alpha_i$  if $\beta=0$,
and   $\gamma= \bigvee_{i=1}^n  \alpha_i$  if $\beta=1$, corresponding to statements~\eqref{eq:elw1} and \eqref{eq:elw2}, respectively.
Thus, it suffices to consider $\beta\in (0,1)$.
To proceed, we need the following lemma, whose proof will be given in  the appendix.

\begin{lemma}\label{lem:rs1}
Let $\beta \in (0,1)$ and $Y\uparrow_\beta X$. Denote  
 $B=A^X_{1-\beta}$.
We have the statements:
\begin{enumerate}[\rm (i)]
\item  $B\subset \{Y \ge  \VaRl_{ \beta}(Y)\}$ and $B^c \subset  \{Y \le  \VaRl_{ \beta}(Y)\}$ a.s.
\item If $\alpha> \beta$,
then $\VaRl_{ \alpha}(Y) = \VaRl_{ \alpha-\beta} ( z\id_B + Y\id_{B^c})$ for all $z\le \VaRl_{ \alpha }(Y) $.
\item If $\alpha> \beta$,
then $\VaRl_{ \alpha}(Y) = \VaRl_{ \alpha} ( z\id_B + Y\id_{B^c})$ for   all $z\ge \VaRl_\alpha(Y)$.
\item If  $\alpha\le \beta$,
then $\VaRl_{ \alpha}(Y) = \VaRl_{ \alpha} ( Y\id_B + z\id_{B^c})$ for all $z \le \VaRl_{ \alpha}(Y)$.
\item If $\alpha+\beta<1$, then $\VaRl_{ \alpha}(Z) \ge \VaRl_{ \alpha+\beta} (z\id_B+ Z\id_{B^c}) $ for all  $Z\in \X$ and $z\in \R$.
\end{enumerate}
\end{lemma}

We can now continue the proof of Theorem~\ref{th:rs}.
Let $\beta\in (0,1)$ and take an arbitrary admissible allocation $(X_1,\dots,X_n)\in \mathbb A_n(X)$ such that $X_i\uparrow_\beta X$, $i=1,\dots,n$.
We need additional notation: $B=A^X_{1-\beta}$, $J=\{i\in \{1,\dots,n\}: \alpha_i> \beta\}$, and $K=\{1,\dots,n\}\setminus J$. 
 Moreover,  let $x_i=\VaRl_{ \alpha_i} (X_i)$, $y_i=\VaRl_{ \beta}(X_i)$, $i=1,\dots,n$,
  $y_J=\sum_{i\in J}y_i$,   $y_K=\sum_{i\in K}y_i$,
  $X_J=\sum_{i\in J}X_i$,   and $X_K=\sum_{i\in K}X_i$.

By Lemma~\ref{lem:rs1}(i), we have  (all statements are in the sense of a.s.)
\begin{equation}\label{eq:rspf-5}
 B\subset \{X_i \ge y_i\}\mbox{~~and~~} B^c \subset  \{X_i \le y_i\} \mbox{~~  for each $i=1,\dots,n$}.
 \end{equation}
Using statements~\eqref{eq:rspf-5}, we see that
 the random vector $(X_i \id_B + y_i\id_{B^c} )_{i\in K}$ is strongly comonotonic, because $(X_1,\dots,X_n)$ is strongly comonotonic on the event $B$ by assumption.
 Hence, using $y_i\le x_i$ for $i\in K$, Lemma~\ref{lem:rs1}(iv) and statement~\eqref{eq:elw2}, we get
\begin{align}
\sum_{i\in K} \VaRl_{ \alpha_i}  (X_i) & = \sum_{i\in K} \VaRl_{ \alpha_i}  (X_i \id_B + y_i\id_{B^c} )\notag
\\
& \ge  \VaRl_{ \bigvee_{i\in K} \alpha_i}  \left(\sum_{i\in K}X_i \id_B + \sum_{i\in K}y_i \id_{B^c} \right)\notag
\\
& =  \VaRl_{ \bigvee_{i\in K} \alpha_i}  \left(X_K \id_B + y_K \id_{B^c} \right). \label{eq:rspf-2}
\end{align}
Further,   statements~\eqref{eq:rspf-5} also imply
\begin{equation}\label{eq:rspf-4}
B \subset \{ X_K  \ge y_K\} ,~~B \subset \{ X_J  \ge y_J\},~~B^c \subset \{X_K \le y_K\} \mbox{~~and~~} B^c \subset \{X_J \le y_J\} .
\end{equation}
We split the following considerations into two cases.

\paragraph{\it Case 1.}
    Assume $\beta\ge \bigvee_{i=1}^n  \alpha_i$, which means $K=\{1,\dots,n\}$ and $\gamma=\bigvee_{i=1}^n \alpha_i$.
Note that statements~\eqref{eq:rspf-4}  imply   $X_K \id_B + y_K \id_{B^c} \ge X_K\ge X$.
Using bound~\eqref{eq:rspf-2} and the fact that $\VaRl_\gamma(Y)$ is increasing in $Y$, we have
 $$  \sum_{i=1}^n   \VaRl_{ \alpha_i}  (X_i)\ge     \VaRl_{ \bigvee_{i=1}^n \alpha_i} \left(X_K \id_B + y_K \id_{B^c} \right) \ge \VaRl_{ \gamma}  \left( X_K\right)\ge \VaRl_{  \gamma }  (X).$$
 Therefore, $V_\beta(X) \ge \VaRl_{  \gamma}  (X)$.
  On the other hand,
by statement~\eqref{eq:elw2}, we have
 $$V_\beta (X)\le \VaRl_{ \bigvee_{i=1}^n  \alpha_i}(X) = \VaRl_{  \gamma }  (X)  .$$
 Putting the above   observations together, we get   $V_\beta(X)   =   \VaRl_{  \gamma}  (X)$.

\paragraph{\it Case 2.}
 Assume $\beta< \bigvee_{i=1}^n  \alpha_i$, which  means $\gamma =\beta+ \sum_{i=1}^n (\alpha_i-\beta)_+ >\beta$, and $J\ne \varnothing$.
 Using Lemma~\ref{lem:rs1}(ii) and (v), and bound~\eqref{eq:elw3},
  we get
\begin{align}
\sum_{i\in J} \VaRl_{ \alpha_i}  (X_i) & = \sum_{i\in J} \VaRl_{ \alpha_i-\beta}  (x_i \id_B + X_i\id_{B^c} ) \notag
\\
& \ge  \VaRl_{ \sum_{i\in J}(\alpha_i-\beta)}  \left(\sum_{i\in J}x_i \id_B + \sum_{i\in J}X_i\id_{B^c} \right) \notag \\
&\ge   \VaRl_{ \beta+\sum_{i\in J}(\alpha_i-\beta)}  \left(y_J \id_B + X_J\id_{B^c} \right)\notag \\
&=  \VaRl_{ \gamma}  \left(y_J \id_B + X_J\id_{B^c} \right). \label{eq:rspf-1}
\end{align}
Therefore,
$y_J\id_B + X_J\id_{B^c} $ and $X_K \id_B + y_K\id_{B^c} $ are strongly comonotonic.
Putting   inequalities~\eqref{eq:rspf-2} and \eqref{eq:rspf-1} together, and  using statements~\eqref{eq:elw2} and \eqref{eq:rspf-4}, we obtain
\begin{align}
\sum_{i=1}^n  \VaRl_{ \alpha_i}  (X_i) & = \sum_{i\in J} \VaRl_{ \alpha_i}  (X_i) + \sum_{i\in K} \VaRl_{ \alpha_i}  (X_i)  \notag
\\
& \ge \VaRl_{ \gamma}  \left(y_J \id_B + X_J \id_{B^c} \right) + \VaRl_{ \gamma} \left(X_K \id_B + y_K\id_{B^c} \right) \notag
\\
& \ge  \VaRl_{ \gamma }  \left((X_K +y_J )\id_B + (X_J+y_K) \id_{B^c}    \right)\notag\\
& \ge  \VaRl_{ \gamma }  \left((y_K +y_J )\id_B + (X_J+X_K) \id_{B^c}    \right)\notag \\
& \ge  \VaRl_{ \gamma }  \left((y_K +y_J )\id_B + X \id_{B^c}    \right).
\label{eq:rspf-3}
\end{align}
Note that $X$ is continuously distributed, implying $\VaRl_{\gamma}(X) < \VaRl_\beta(X)$.
Moreover,  $y_K +y_J \ge X_J+X_K\ge X$ on $B^c$, and
by Lemma~\ref{lem:rs1}(i),  we have
$$\{X\le \VaRl_\gamma(X)\}  \subset \{X<\VaRl_\beta(X)\} \subset B^c\subset \{X \le y_K +y_J  \}. $$
This shows
$y_K+y_J \ge  \VaRl_\gamma (X) $.
Using Lemma~\ref{lem:rs1}(iii) and bounds~\eqref{eq:rspf-3}, we obtain
$$
\sum_{i=1}^n  \VaRl_{ \alpha_i}  (X_i)   \ge  \VaRl_{ \gamma }  \left((y_K +y_J )\id_B + X \id_{B^c}    \right)
 =   \VaRl_{ \gamma} (X) .$$
This proves $ V_\beta(X)\ge \VaRl_{ \gamma}(X)$.

Next, we show
$ V_\beta(X)\le \VaRl_{ \gamma} (X)$ by an explicit construction of an optimal allocation.
Let $y=\VaRl_{ \beta}(X)$ and $z=\VaRl_{ \gamma}(X)$.
Without loss of generality, assume $1\in J$.
Recall that $$\p(A_{1-\gamma}^X\setminus A_{1-\beta}^X)  = \gamma-\beta =\sum_{i\in J}(\alpha_i-\beta),$$
and hence we can find
  a partition    $(A_i)_{i\in J}$  of $A_{1-\gamma}^X\setminus A_{1-\beta}^X$
such that $\p(A_i)=\alpha_i-\beta_i$ for each $i\in J$.
Define
\begin{align} \label{eq:opt}
X_i=
\left\{
\begin{array}{ll}
 (X-z) \left (\id_B   + \id_{A_1} +  \id_{(A_{1-\gamma}^X)^c}\right) + z & \mbox{if } i=1,\\
  y_+\id_B + (X-z) \id_{A_i} & \mbox{if }i\in J\setminus\{1\},\\
 0 &\mbox{if }  i\in K.
\end{array}
 \right.
 \end{align}
We easily verify that
$\sum_{i=1}^n X_i = (\#J-1) y_+\id_B + X\ge X$
and $X_i\uparrow_{\beta} X$, $i=1,\dots,n$.
Hence, $(X_1,\dots,X_n)$ is an admissible allocation for problem~\eqref{eq:rs4}.
Furthermore, we check that
$\VaRl_{ \alpha_1}(X)=z$
and
$\VaRl_{ \alpha_i}(X_i)=0$ for $i\ne 1$.
Therefore, $$ \sum_{i=1}^n \VaRl_{ \alpha_i}(X_i)=z =  \VaRl_{ \gamma} (X),$$ showing that $ V_\beta(X)\le  \VaRl_{ \gamma} (X)$.

With this, we finish the proof of Theorem~\ref{th:rs}.
\end{proof}

An explicit construction of an optimal allocation to problem~\eqref{eq:rs4} has been obtained in the proof of Theorem~\ref{th:rs}.
Specifically, and without loss of generality, let  $\alpha_1=   \bigvee_{i=1}^n \alpha_i$.
If $\beta<\bigvee_{i=1}^n \alpha_i$, then an optimal allocation is given by equation~\eqref{eq:opt}.
On the other hand, if $\beta \ge \bigvee_{i=1}^n \alpha_i$,
then an optimal allocation is trivially given by $X_1=X$ and $X_i=0$ for $i\ne 1$.
The optimal allocations are generally not unique, similarly to the case of problems~\eqref{eq:rs2} and~\eqref{eq:rs3} in \cite{ELW18}.

Finally, we discuss the implication of the values of the parameter $\beta $ in problem \eqref{eq:rs4}.
Recall that $V_\beta(X)$ represents the smallest total risk measure after risk redistribution.
In Theorem \ref{th:rs}, $\gamma=\gamma(\beta)$ is a piece-wise linear decreasing function of $\beta$,
with  $\gamma(0) = \sum_{i=1}^n \alpha_i$
and   $\gamma(\beta)= \bigvee_{i=1}^n  \alpha_i$  if $\beta\ge \bigvee_{i=1}^n  \alpha_i$.
Thus, if there is no dependence constraint, we arrive at \eqref{eq:elw1}, the minimum possible total risk measure obtained by \citet[Theorem 2]{ELW18}.
If the dependence constraint is strong enough (i.e., $\beta\ge \bigvee_{i=1}^n  \alpha_i$), then
we arrive at   the same value of the minimum total risk measure to \eqref{eq:elw2}, obtained by \citet[Proposition 5]{ELW18}.
If the dependence constraint is intermediate, then the total risk measure $V_\beta(X)$ varies between the two values, decreasing in $\beta$.
This suggests that the use of weak comonotonicity as a dependence constraint yields a spectrum of flexible formulations of the risk sharing problem.

\section{Summary and concluding notes}
\label{sect-6}

In this paper, we   introduced the notion of weak comonotonicity.
Via the analysis of several properties and applications, we   show the encompassing nature of weak comonotonicity, which contains -- as a special case -- the classical notion of comonotonicity. The new notion serves a bridge that connects the classical notion of comonotonicity of random variables with a number of well-known notions of (in)dependence and association \cite[e.g.,][and references therein]{J14, DS15}.  More importantly, we illustrate that introduced weak comonotonicity provides necessary and sufficient conditions for a number of problems in economics, banking, and insurance, and in particular to those dealing with risk aggregation and risk sharing. Specifically,  it is shown that the notion of weak comonotonicity yields a sufficient condition for the maximum $\VaR$ aggregation, and a necessary and sufficient condition for the maximum $\ES$ aggregation. As far as we are aware of, such conditions have been elusive. In addition,  we    provided analytical solutions to a risk sharing problem whose constraint on the dependence structure of admissible allocations has been most naturally described by weak comonotonicity, bridging the gap between strong comonotonicity and no dependence assumption studied in the literature.
We finally remark that, as weak comonotonicity depends on the set $\mathcal P$ of product measures,  its spectrum  is very wide,  including many types of dependence.

\section*{Acknowledgements}

The authors thank the Editor  Emanuele Borgonovo  and four anonymous referees for various helpful comments on an early version of the paper.
The authors have been supported by their individual research grants from the Natural Sciences and Engineering Research Council (NSERC) of Canada (RGPIN-2018-03823, RGPAS-2018-522590, RGPIN-2016-427216), as well as by the National Research Organization ``Mathematics of Information Technology and Complex Systems'' (MITACS) of Canada.

\appendix

\section{Appendix: Proof of Lemma~\ref{lem:rs1}}
\label{app:lemrs}

\paragraph{\it Proof of statement (i).}
By definition of $Y\uparrow_\beta X$, for a.s.~all $\omega\in B$ and $\omega'\in B^c$, we have $Y(\omega)\ge Y(\omega')$.
Therefore, there exists a constant $z\in \R$ such that
$B\subset \{Y \ge z\}$ and $B^c \subset  \{Y \le z\}$ a.s.
It is easy to see that this constant can be chosen as $z= \VaRl_{ \beta}(Y)$ because $\p(Y\ge  \VaRl_{ \beta}(Y) )\ge 1-\beta$ and $\p(Y\le  \VaRl_{ \beta}(Y))\ge \beta$.

\paragraph{\it Proof of statement (ii).}

By bound~\eqref{eq:elw3}, we have
\begin{align*}
\VaRl_{\alpha-\beta} (z \id _B + Y\id_{B^c}) + \VaRl_{\beta} ((Y-z) \id _B) \ge \VaRl_\alpha(Y).
\end{align*}
Note that $\VaRl_{\beta} ((Y-z) \id _B) = 0$ since $Y\ge z$ on $B$ by statement~(i) and $\p(B)=\beta$.
Hence,
\begin{align} \label{eq:pflm-1}
\VaRl_{\alpha-\beta} (z \id _B + Y\id_{B^c})
\ge  \VaRl_\alpha(Y).
\end{align}
To show the other direction, we consider two cases.
If $\VaRl_\alpha(Y)< \VaRl_\beta(Y)$, then $\{Y  \le  \VaRl_\alpha(Y)\}\subset \{Y  < \VaRl_\beta(Y) \}\subset B^c $ by statement~(i).
In this case,
\begin{align*}
\p(z\id_B + Y\id_{B^c} \le\VaRl_\alpha(Y)  ) & = \p(z\le \VaRl_\alpha(Y), ~B) + \p( Y  \le \VaRl_\alpha(Y) ,~ B^c)
\\&=  \p(B) + \p(Y\le \VaRl_\alpha(Y))
\\ &\ge \beta + 1-\alpha.
\end{align*}
If $\VaRl_\alpha(Y)= \VaRl_\beta(Y)$, then $B^c \subset \{Y  \le  \VaRl_\beta(Y)\} = \{Y  \le  \VaRl_\alpha(Y)\}$ by statement~(i).
In this case,
\begin{align*}
\p(z\id_B + Y\id_{B^c} \le\VaRl_\alpha(Y)  ) & = \p(z\le \VaRl_\alpha(Y), ~B) + \p( Y  \le \VaRl_\alpha(Y) ,~ B^c)
\\&=  \p(B) + \p(B^c)  =1.
\end{align*} In both cases,
$$\p(z\id_B + Y\id_{B^c} \le\VaRl_\alpha(Y)  )  \ge  1-(\alpha -  \beta),$$
 which implies $ \VaRl_{ \alpha-\beta} ( z\id_B + Y\id_{B^c})\le \VaRl_{ \alpha}(Y)$.
By bound~\eqref{eq:pflm-1}, we get $ \VaRl_{ \alpha-\beta} ( z\id_B + Y\id_{B^c})= \VaRl_{ \alpha}(Y)$.

\paragraph{\it Proof of statement (iii).}
Note that
\begin{align*}
\p(z\id_B + Y\id_{B^c} \ge \VaRl_\alpha(Y)  ) & = \p(z\ge  \VaRl_\alpha(Y), ~B) + \p( Y  \ge \VaRl_\alpha(Y) ,~ B^c)
\\&= \p(B) + \p( Y  \ge \VaRl_\alpha(Y) ,~ B^c)
\\ &\ge \p(Y  \ge \VaRl_\alpha(Y) ,~B) + \p( Y  \ge \VaRl_\alpha(Y) ,~ B^c)
\\&= \p(Y\ge \VaRl_\alpha(Y))
\ge  \alpha.
\end{align*}
This shows
\begin{align} \label{eq:pflm-2}\VaRl_{\alpha} (z \id _B + Y\id_{B^c})
\ge  \VaRl_\alpha(Y).\end{align}
For the other direction, we consider two cases, similarly to statement~(ii).
If $\VaRl_\alpha(Y)< \VaRl_\beta(Y)$, then $\{Y  \le  \VaRl_\alpha(Y)\}\subset \{Y  < \VaRl_\beta(Y) \}\subset B^c $ by statement~(i).
In this case,
\begin{align*}
\p(z\id_B + Y\id_{B^c} \le\VaRl_\alpha(Y)  ) & = \p(z\le \VaRl_\alpha(Y), ~B) + \p( Y  \le \VaRl_\alpha(Y) ,~ B^c)
\\&\ge   \p(Y\le \VaRl_\alpha(Y))
\ge   1-\alpha.
\end{align*}
If $\VaRl_\alpha(Y)= \VaRl_\beta(Y)$, then $B^c \subset \{Y  \le  \VaRl_\beta(Y)\} = \{Y  \le  \VaRl_\alpha(Y)\}$ by statement~(i).
In this case,
\begin{align*}
\p(z\id_B + Y\id_{B^c} \le\VaRl_\alpha(Y)  ) & = \p(z\le \VaRl_\alpha(Y), ~B) + \p( Y  \le \VaRl_\alpha(Y) ,~ B^c)
\\&\ge  \p(B^c)  =1-\beta \ge 1-\alpha.
\end{align*} In both cases,
$$\p(z\id_B + Y\id_{B^c} \le\VaRl_\alpha(Y)  )  \ge  1-\alpha$$
 which implies $ \VaRl_{ \alpha} ( z\id_B + Y\id_{B^c})\le \VaRl_{ \alpha}(Y)$.
By bound~\eqref{eq:pflm-2}, we get $ \VaRl_{ \alpha} ( z\id_B + Y\id_{B^c})= \VaRl_{ \alpha}(Y)$.

\paragraph{\it Proof of statement (iv).}
If  $\alpha\le \beta$,
then $\VaRl_{ \alpha}(Y) = \VaRl_{ \alpha} ( Y\id_B + z\id_{B^c})$ for all $z \le y$. Note that
\begin{align*}
\p(Y\id_B + z\id_{B^c} \le \VaRl_\alpha(Y)  ) & = \p(Y\le   \VaRl_\alpha(Y), ~B) + \p( z  \le  \VaRl_\alpha(Y) ,~ B^c)
\\  & = \p(Y\le   \VaRl_\alpha(Y), ~B) + \p(   B^c)
\\ & \ge  \p(Y\le   \VaRl_\alpha(Y))  \ge 1-\alpha.
\end{align*}
This shows
\begin{align} \label{eq:pflm-3}\VaRl_{\alpha} (Y \id _B + z\id_{B^c})
\le  \VaRl_\alpha(Y).\end{align}
For the other direction, we again consider two cases.
If $\VaRl_\alpha(Y)> \VaRl_\beta(Y)$, then $\{Y  \ge  \VaRl_\alpha(Y)\}\subset \{Y  > \VaRl_\beta(Y) \}\subset B $ by statement~(i).
In this case,
\begin{align*}
\p(Y\id_B + z\id_{B^c} \ge \VaRl_\alpha(Y)   ) & = \p(Y \ge \VaRl_\alpha(Y), ~B) + \p( z \ge \VaRl_\alpha(Y) ,~ B^c)
\\&\ge   \p(Y\ge  \VaRl_\alpha(Y))
\ge   \alpha.
\end{align*}
If $\VaRl_\alpha(Y)= \VaRl_\beta(Y)$, then $B \subset \{Y  \ge  \VaRl_\beta(Y)\} = \{Y  \ge  \VaRl_\alpha(Y)\}$ by statement~(i).
In this case,
\begin{align*}
\p(Y\id_B + z\id_{B^c} \ge \VaRl_\alpha(Y)  ) & = \p(Y \ge \VaRl_\alpha(Y), ~B) + \p( z  \ge \VaRl_\alpha(Y) ,~ B^c)
\\&\ge  \p(B)  =\beta \ge \alpha.
\end{align*} In both cases,
$$\p(Y\id_B + z\id_{B^c} \ge\VaRl_\alpha(Y)  )  \ge  \alpha$$
 which implies $ \VaRl_{ \alpha} ( Y\id_B + z\id_{B^c})\ge  \VaRl_{ \alpha}(Y)$.
By bound~\eqref{eq:pflm-3}, we get $ \VaRl_{ \alpha} ( Y\id_B + z\id_{B^c})= \VaRl_{ \alpha}(Y)$.

\paragraph{\it Proof of statement (v).}
Using~\eqref{eq:elw3}, we have
 $\VaRl_{ \alpha}(Z) +\VaRl_\beta((Z-z)\id_B) \ge \VaRl_{ \alpha+\beta} (z\id_B+ Z\id_{B^c}) $.
Note that $\VaRl_\beta((Z-z)\id_B)\le 0$ since  $\p((Z-z)\id_B\ge 0)\le \p(B) =1-\beta$.
 Hence $\VaRl_{ \alpha}(Z) \ge \VaRl_{ \alpha+\beta} (z\id_B+ Z\id_{B^c}) $.  This establishes statement~(v) and concludes the entire proof of Lemma~\ref{lem:rs1} \qed


\begin{thebibliography}{10}
\small



\bibitem[\protect\citeauthoryear{Box et al.}{Box et al.}{2015}]{BJRL15}
Box, G.E.P., Jenkins, G.M., Reinsel, G.C. and Ljung, G.M. (2015).
{\it Time Series Analysis: Forecasting and Control.} (Fifth edition.)
Wiley, New York.




\bibitem[\protect\citeauthoryear{Carlier et al.}{Carlier et al.}{2012}]{CDG12}
 Carlier, G.,  Dana, R.-A. and Galichon, A. (2012). Pareto efficiency for the concave order and multivariate
comonotonicity. \emph{Journal of Economic Theory}, \textbf{147}, 207--229.




\bibitem[\protect\citeauthoryear{Cheung}{2010}]{C10}
Cheung, K.~C. (2010).
Characterizing a comonotonic random vector by the distribution of the sum
of its components. \emph{Insurance: Mathematics and Economics}, \textbf{47}, 130--136.


\bibitem[\protect\citeauthoryear{Dhaene et al.}{Dhaene et al.}{2002a}]{DDGV02a}
Dhaene, J., Denuit, M., Goovaerts, M.J. and Vyncke, D. (2002a).
The concept of comonotonicity in actuarial science and finance: theory.
\textit{Insurance: Mathematics and Economics}, \textbf{31}, 3--33.

\bibitem[\protect\citeauthoryear{Dhaene et al.}{Dhaene et al.}{2002b}]{DDGV02b}
Dhaene, J., Denuit, M., Goovaerts, M.J. and Vyncke, D. (2002b).
The concept of comonotonicity in actuarial science and finance: applications.  \textit{Insurance: Mathematics and Economics}, \textbf{31}, 133--161.

\bibitem[\protect\citeauthoryear{Denneberg}{Denneberg}{1994}]{D94}
Denneberg, D. (1994).
{\em Non-additive Measure and Integral}. Kluwer, Dordrecht.






\bibitem[\protect\citeauthoryear{Deprez and Gerber}{Deprez and Gerber}{1985}]{DG85}
{Deprez, O. and Gerber, H.U.} (1985). On convex principles of premium calculation. \emph{Insurance: Mathematics and Economics}, \textbf{4}(3), 179--189.




\bibitem[\protect\citeauthoryear{Durante and Sempi}{Durante and Sempi}{2015}]{DS15}
Durante, F. and Sempi, C. (2015).
\textit{Principles of Copula Theory.}
Chapman and Hall/CRC, Boca Raton, FL.





\bibitem[\protect\citeauthoryear{Ekland et al.}{2012}]{EGH12}
Ekeland, I., Galichon, A.  and Henry, M. (2012). Comonotonic measures of multivariate risks. \emph{Mathematical Finance}, \textbf{22}, 109--132.



\bibitem[\protect\citeauthoryear{Embrechts et al.}{2018}]{ELW18}
Embrechts, P., Liu, H. and Wang, R. (2018).
Quantile-based risk sharing.
\emph{Operations Research}, \textbf{66}, 936--949.

\bibitem[\protect\citeauthoryear{Embrechts et al.}{2019}]{ELMW19}
Embrechts, P., Liu, H., Mao, T. and Wang, R. (2019).
Quantile-based risk sharing with heterogeneous beliefs.
\emph{Mathematical Programming Series B} (in press).
 https://doi.org/10.1007/s10107-018-1313-1


\bibitem[\protect\citeauthoryear{Embrechts et al.}{Embrechts et al.}{2013}]{EPR13}  Embrechts, P.,  Puccetti, G. and R\"{u}schendorf, L.  (2013). Model uncertainty and VaR aggregation. {\em Journal of Banking and Finance}, \textbf{37}, 2750--2764.


\bibitem[\protect\citeauthoryear{Embrechts et al.}{Embrechts et al.}{2014}]{EPRWB14} {Embrechts, P., Puccetti, G., R\"uschendorf, L., Wang, R. and Beleraj, A.} (2014). An academic response to Basel 3.5. \emph{Risks}, \textbf{2}, 25-48.



\bibitem[\protect\citeauthoryear{Embrechts et al.}{Embrechts et al.}{2015}]{EWW15} {Embrechts, P., Wang, B. and Wang, R.} (2015). Aggregation-robustness and model uncertainty of regulatory risk measures.  {\em Finance and Stochastics},  \textbf{19}, 763--790.


\bibitem[\protect\citeauthoryear{Embrechts and Wang}{Embrechts and Wang}{2015}]{EW15} {Embrechts, P. and Wang, R.} (2015). Seven proofs for the subadditivity of expected shortfall.  \emph{Dependence Modeling}, \textbf{3}, 126--140.

\bibitem[\protect\citeauthoryear{Engle}{Engle}{2016}]{E16}
Engle, R.F. (2016).
Dynamic conditional beta.
\textit{Journal of Financial Econometrics}, \textbf{14}, 643--667.

\bibitem[\protect\citeauthoryear{Esary et al.}{Esary et al.}{1967}]{EPW67}
Esary, J.D., Proschan, F., and Walkup, D.W. (1967).
Association of random variables, with applications.
\textit{Annals of Mathematical Statistics}, \textbf{38}, 1466--1474.


\bibitem[\protect\citeauthoryear{F\"{o}llmer and Schied}{2016}]{FS16}
{F\"{o}llmer, H. and Schied, A.} (2016).
\emph{Stochastic Finance: An Introduction in Discrete Time.} (Fourth Edition.) Walter de Gruyter, Berlin.


\bibitem[\protect\citeauthoryear{Furman and Zitikis}{Furman and Zitikis}{2009}]{FZ09}
Furman, E.\, and Zitikis, R.\, (2009).
Weighted pricing functionals with applications to insurance: an overview.
\textit{North American Actuarial Journal}, \textbf{13}, 483--496.


\bibitem[\protect\citeauthoryear{Furman et al.}{Furman et al.}{2017}]{FWZ17}
{Furman, E., Wang, R. and Zitikis, R.} (2017). Gini-type measures of risk and variability: Gini shortfall, capital allocation and heavy-tailed risks. \emph{Journal of Banking and Finance}, \textbf{83}, 70--84.


\bibitem[\protect\citeauthoryear{Gebelein}{Gebelein}{1941}]{G41}
Gebelein, H. (1941).
Das statistische problem der korrelation als variations- und eigenwertproblem und sein zusammenhang mit der ausgleichsrechnung.  \textit{Zeitschrift f\"{u}r Angewandte Mathematik und Mechanik}, \textbf{21}, 364--379.


\bibitem[\protect\citeauthoryear{Gillen and Markowitz}{Gillen and Markowitz}{2009}]{GM09}
Gillen, B. and Markowitz, H.M. (2009). A taxonomy of utility functions.
In: \textit{Variations in Economic Analysis: Essays in Honor of Eli Schwartz}
(Eds.: J.R. Aronson, H.L. Parmet, and R.J. Thornton). Springer, New York.



\bibitem[\protect\citeauthoryear{Joe}{Joe}{1997}]{J97}
Joe, H. (1997).
\textit{Multivariate Models and Multivariate Dependence Concepts.}
Springer, Dordrecht.



\bibitem[\protect\citeauthoryear{Joe}{Joe}{2014}]{J14}
Joe, H. (2014).
\textit{Dependence Modeling with Copulas.}
Chapman and Hall/CRC, Boca Raton, FL.


\bibitem[\protect\citeauthoryear{Kimeldorf and Sampson}{Kimeldorf and Sampson}{1978}]{KS78}
Kimeldorf, G. and Sampson, A.R. (1978).
Monotone dependence. \textit{Annals of Statistics}, \textbf{6}, 895--903.

\bibitem[\protect\citeauthoryear{Lehmann}{Lehmann}{1966}]{L66}
Lehmann, E.L. (1966).
Some concepts of dependence.
\textit{Annals of Mathematical Statistics}, \textbf{37}, 1137--1153.


\bibitem[\protect\citeauthoryear{Makarov}{Makarov}{1981}]{M81}
Makarov, G.~D. (1981).
  Estimates for the distribution function of the sum of two random
  variables with given marginal distributions.
  {\em Theory of Probability and its Applications},~{\bf 26}, 803--806.


\bibitem[\protect\citeauthoryear{Markowitz}{Markowitz}{1952}]{M52}
Markowitz, H. (1952). The utility of wealth.
\textit{Journal of Political Economy}, \textbf{60}, 151--156.


\bibitem[\protect\citeauthoryear{McNeil et al.}{McNeil et al.}{2015}]{MFE15}
{McNeil, A. J., Frey, R. and Embrechts, P.} (2015). \emph{Quantitative
Risk Management: Concepts, Techniques and Tools}. Revised Edition.  Princeton, NJ:
Princeton University Press.



\bibitem[\protect\citeauthoryear{Pennings and Smidts}{Pennings and Smidts}{2003}]{PS03}
Pennings, J.M.E., and Smidts, A. (2003).
The shape of utility functions and organizational behavior.
\textit{Management Science}, \textbf{49}, 1251--1263.


\bibitem[\protect\citeauthoryear{Pflug and R\"{o}misch}{Pflug and R\"{o}misch}{2007}]{PR07}
Pflug, G.C. and R\"{o}misch, W. (2007).
\textit{Modelling, Managing and Measuring Risks.}
World Scientific Publishing,
Singapore.

\bibitem[\protect\citeauthoryear{Puccetti and Scarsini}{Puccetti and Scarsini}{2010}]{PS10}
Puccetti, G. and Scarsini, M. (2010).
Multivariate comonotonicity.
\textit{Journal of Multivariate Analysis}, \textbf{101}, 291--304.



\bibitem[\protect\citeauthoryear{Puccetti and Wang}{Puccetti and
  Wang}{2015}]{PW15}
Puccetti, G. and Wang R.  (2015).
Extremal dependence concepts.
 \emph{Statistical Science},  \textbf{30}, 485--517.

\bibitem[\protect\citeauthoryear{Rao}{Rao}{1997}]{R97}
Rao, C.R. (1997).
\textit{Statistics and Truth: Putting Chance to Work.}
World Scientific, Singapore.

\bibitem[\protect\citeauthoryear{R\"uschendorf}{R\"uschendorf}{1982}]{R82} {R\"uschendorf, L.} (1982). Random variables with maximum sums. \emph{Advances in Applied Probability}, \textbf{14},
623--632.



\bibitem[\protect\citeauthoryear{R{\"u}schendorf}{R{\"u}schendorf}{2013}]{R13}
R{\"u}schendorf, L. (2013).
  {\em Mathematical Risk Analysis. Dependence, Risk Bounds, Optimal
  Allocations and Portfolios}.
  Springer, Heidelberg.


\bibitem[\protect\citeauthoryear{Schmeidler}{Schmeidler}{1986}]{S86}
Schmeidler, D. (1986).
Integral representation without additivity.
{\it Proceedings of the American Mathematical Society}, \textbf{97}, 255--261.


\bibitem[\protect\citeauthoryear{Wang, Peng and Yang}{Wang
  et~al.}{2013}]{WPY13}
Wang, R., Peng, L. and Yang, J. (2013).
 Bounds for the sum of dependent risks and worst Value-at-Risk with
  monotone marginal densities.
 {\em Finance and Stochastics}, \textbf{17}, 395--417.


\bibitem[\protect\citeauthoryear{Yaari}{Yaari}{1987}]{Y87}
Yaari, M.E. (1987).
The dual theory of choice under risk.
\textit{Econometrica},  \textbf{55}, 95--115.


\end{thebibliography}
\end{document}